\documentclass[preprint,aps,prd,floatfix]{revtex4-1}
\usepackage{graphicx}
\usepackage{amsfonts}
\usepackage{amsmath}
\usepackage{rotating}
\usepackage{amssymb,amsthm}

\newcommand\GLC{{\overset{\ \circ}{G}}{}}

\newtheorem{thm}{Theorem}[section] 
\newtheorem{cor}[thm]{Corollary} 
 
\newtheorem{prop}[thm]{Proposition} 
 
\theoremstyle{definition} 
  
\theoremstyle{remark}

\newcommand{\mbold}[1]{\mbox{\boldmath{\ensuremath{#1}}}}

\def\beq{\begin{eqnarray}}  
\def\eeq{\end{eqnarray}}  

\def \bell {\mbox{{\mbold\ell}}}
\def \bn {\mbox{{\bf n}}}
\def \bm {\mbox{{\bf m}}}
\def \bh {\mbox{{\bf h}}}
\def \bomega {\mbox{{\mbold \omega}}}

\def \bT {\mbox{{ {\bf T}}}}
\def \bV {\mbox{{{\bf V}}}}
\def \bA {\mbox{{{\bf A}}}}

\begin{document}


\title{Symmetry and Equivalence in Teleparallel Gravity}

\author {A. A. Coley}
\email{aac@mathstat.dal.ca}
\affiliation{Department of Mathematics and Statistics, Dalhousie University, Halifax, Nova Scotia, Canada, B3H 3J5}

\author{R. J. \surname{van den Hoogen}}
\email{rvandenh@stfx.ca}
\affiliation{Department of Mathematics and Statistics, St. Francis Xavier University, Antigonish, Nova Scotia, Canada, B2G 2W5}

\author {D. D. McNutt}
\email{david.d.mcnutt@uis.no}
\affiliation{ Department of Mathematics and Physics, University of Stavanger, Stavanger, Norway}



\begin{abstract}

In theories such as teleparallel gravity and its extensions, the frame basis replaces the metric tensor as the primary object of study. A choice of coordinate system, frame basis and spin-connection must be made to obtain a solution from the field equations of a given teleparallel gravity theory. It is worthwhile to express solutions in an invariant manner in terms of torsion invariants to distinguish between different solutions. In this paper we discuss the symmetries of teleparallel gravity theories, describe the classification of the torsion tensor and its covariant derivative and define scalar invariants in terms of the torsion. In particular, we propose a modification of the Cartan-Karlhede algorithm for geometries with torsion (and no curvature or nonmetricity). The algorithm determines the dimension of the symmetry group for a solution and suggests an alternative frame-based approach to calculating symmetries. We prove that the only maximally symmetric solution to any theory of gravitation admitting a non-zero torsion tensor is Minkowski space. As an illustration we apply the algorithm to six particular exact teleparallel geometries. From these examples we notice that the symmetry group of the solutions of a teleparallel gravity theory is potentially smaller than their metric-based analogues in General Relativity.



%

\end{abstract}

\maketitle

\tableofcontents 

\section{Introduction}

While Einstein's General Theory of Relativity (GR) is well accepted, there continues to be interest in alternative theories. It can be argued that GR has been verified for isolated masses on length scales of the solar system but that it continues to face challenges on both quantum and cosmological scales.  One notable problem is the inclusion of  dark energy and dark matter \cite{Nojiri_Odintsov2006,Capozziello_DeLaurentis_2011} in cosmological models. These quantities play important roles in our description of the universe: dark matter influences the large scale dynamics of matter and dark energy is associated with the observed accelerated expansion of the universe.

Dark energy can be introduced into GR by treating the cosmological constant, $\Lambda$, as an additional parameter in physics  (e.g.,  by treating $\Lambda$ as the vacuum expectation value). However, this vacuum expectation value calculated via quantum field theory is 120 orders of magnitude larger than the observed value. Furthermore, GR is unable to describe spacetime physics in the quantum regime: i.e., at scales of the order of the Planck length. Einstein-Cartan theories, where both torsion and curvature are non-zero, have been proposed to describe gravity in the quantum regime \cite{sharma2014}. One particular class of gravity theories assumes that the dynamics of the gravitational field are encoded in the torsion of spacetime with vanishing curvature \cite{Li_Sotiriou_Barrow2010, Krssak2015, Bahamonde_Boehmer_Wright2015}.

A notable class of torsionful but curvatureless gravitational theories of gravity, called teleparallel theories of gravity (or teleparallel gravity in brief) arise from assuming that both the non-metricity and the curvature of the affine connection are zero. Teleparallel theories of gravity have a long history of analysis, including Einstein himself who believed that the space of distant parallelism, (also called ``absolute parallelism'' or ``tele-parallelism'' by other authors) is the most promising candidate to unify gravitation and electromagnetism \cite{Vargas1992}. Intriguingly, there exists a subclass of teleparallel theories of gravity that are dynamically equivalent to GR.  The Lagrangian for this teleparallel equivalent to GR (TEGR) theories is based on a scalar, $T$, constructed from the torsion tensor and differs from the Lagrangian of GR by a total derivative, implying the field equations for each are formally equivalent. We will call the solutions of the field equations to teleparallel gravity theories, {\em teleparallel geometries}.

One popular generalization of the TEGR is $f(T)$ gravity. In one formulation of $f(T)$ gravity \cite{Li_Sotiriou_Barrow2010, Ferraro:2006jd, Bengochea:2008gz, Sotiriou_Li_Barrow2011,Ferraro_Fiorini2011} (called {\it pure tetrad teleparallel $f(T)$ gravity}) the spin-connection is required to vanish with respect to all frames leading to the torsion tensor being effectively replaced with the coefficients of anholonomy, which is not covariant under local Lorentz transformations. This implies a violation of Lorentz symmetry, so that these theories are frame dependent, as any solution to the field equations depend on the choice of frame \cite{Li_Sotiriou_Barrow2010}. 
To reconcile this, it has been noted that if the teleparallel geometry is defined in a gauge invariant manner as a geometry with zero curvature, then the most general spin-connection that satisfies this requirement is the {\it purely inertial spin-connection}, which vanishes in a very special class of frames (``proper frames'') where all inertial effects are absent, and non-zero in all other frames \cite{Obukhov_Pereira2003,Obukhov_Rubilar2006,Lucas_Obukhov_Pereira2009,Aldrovandi_Pereira2013,Krssak:2018ywd}. The primary benefit of this approach is that by using the purely inertial connection, the resulting teleparallel gravity theory embodied by the Lorentz covariant field equations (We refer to this theory as {\it covariant teleparallel $f(T)$ gravity}) has local Lorentz invariance. With this approach, which we take here, one may use an arbitrary tetrad in an arbitrary coordinate system with the corresponding spin-connection to produce equivalent field equations to those in the proper frame \cite{Krssak_Saridakis2015,Krssak_Pereira2015}.

For a given teleparallel gravity theory, to determine solutions of the corresponding field equations one must choose a coordinate system, $x^\mu$, a frame basis, $\bh^a$, and a spin-connection, $\omega^a_{~bc}$ or alternatively, a coordinate system and a proper frame basis in which the spin-connection is trivial. The choice of frame and the resulting spin-connection will be discussed in section \ref{sec:BasicTEGR}. It is possible that two seemingly distinct choices of the coordinates, frame basis and spin-connection which satisfy the teleparallel field equations are in fact the same solution but this fact can be hidden by these choices. A criteria to uniquely characterize a solution in an invariant manner is necessary. Two solutions can be shown to be inequivalent by comparing the group of symmetries for each solution, such as {\it isometries} which are symmetries of the metric and are necessarily coordinate-independent \cite{Hecht:1992xn}.

The role of an isometry in covariant teleparallel gravity theories is not as clear as in metric based theories. In a teleparallel geometry, the tetrad (or (co)frame) and spin-connection, replace the metric as the principal object of study since the frame and spin-connection are used in the calculation of the torsion tensor and the field equations. The metric is now a derived tensor constructed from symmetric products of the frame elements. It is worthwhile to ask whether the symmetries of a particular teleparallel geometry coincide with the set of isometries.

We note that in metric-based theories, such as those arising from Riemann-Cartan geometries, that in addition to leaving the metric invariant, an isometry is a diffeomorphism from the space into itself and hence preserves the geometric structure of the space. This structure will be discussed in section \ref{sec:CKalg} when the modified Cartan-Karlhede algorithm is reviewed. However, the diffeomorphism condition of an isometry can be stated as: 
\beq \mathcal{L}_{{\bf X}} g_{ab} = 0,~ \mathcal{L}_{{\bf X}} T_{abc}=0,~ \mathcal{L}_{{\bf X}} R_{abcd}=0, \mathcal{L}_{{\bf X}} R_{abcd|e_1 \ldots e_p} = 0 \text{ and } \mathcal{L}_{{\bf X}} T_{abc|e_1 \ldots e_p} = 0  ~\forall p \geq 1. \label{Intro:FS0} \eeq

\noindent In the case of Riemannian geometries, where the Levi-Civita connection is used, the Lie derivatives of the torsion and its covariant derivatives are zero and the resulting equations are trivially satisfied.

For frame based theories, equation \eqref{Intro:FS0} can be restated: 
\beq \mathcal{L}_{{\bf X}} \bh_a = 0,~ \mathcal{L}_{{\bf X}} T_{abc}=0,~ \mathcal{L}_{{\bf X}} R_{abcd}=0, \mathcal{L}_{{\bf X}} R_{abcd;e_1 \ldots e_p} = 0 \text{ and } \mathcal{L}_{{\bf X}} T_{abc;e_1 \ldots e_p} = 0  ~\forall p \geq 1 \label{Intro:FS1} \eeq
\noindent where $\bh_a$ is a geometrically preferred frame known as an {\it invariant frame}. The determination of an invariant frame may require considering the original manifold as a submanifold of a larger manifold and this will be discussed in section \ref{sec:CKalg} as well.  Any vector-field, ${\bf X}$, satisfying \eqref{Intro:FS1} is a {\it frame symmetry} on $\mathcal{F}(M)$, the frame bundle of $M$ \cite{olver1995}. 

We will define an affine frame symmetry as a vector field, ${\bf X}$, satisfying \cite{fon1992}:
\beq \mathcal{L}_{{\bf X}} \bh_a = 0 \text{ and } \mathcal{L}_{{\bf X}} \omega^a_{~bc} = 0, \label{Intro:FS2} \eeq

\noindent where $\omega^a_{~bc}$ denotes the spin-connection relative to the invariant frame. This definition is {\it a frame-dependent} analogue of the definition of a symmetry introduced by \cite{HJKP2018} since the transformation rule for the Lie derivative under a Lorentz transformation, 
%
$\bh_a = \Lambda_{a}^{~b} \bh^{'}_b = (\Lambda^{-1})^b_{~a} \bh_b$ is:
\beq \mathcal{L}_{\bf X} \bh_b & = \mathcal{L}_{\bf X'} \bh^{'}_{d} (\Lambda^{-1})^d_{~b} + (X^{'f} \bh^{'}_d [(\Lambda^{-1})^g_{~e}] \Lambda^e_{~f} \bh^{'}_{g}) (\Lambda^{-1})^d_{~b}. \eeq 

\noindent It can be shown that \eqref{Intro:FS2} is  equivalent to \eqref{Intro:FS1}. These vector fields were previously called {\it affine symmetries} \cite{fon1992}. An affine frame symmetry is an isometry but not all isometries are affine frame symmetries.

Even if two teleparallel geometries have equivalent symmetry groups this will not prove that they are equivalent. To show that two teleparallel geometries are equivalent we must compute and compare invariant quantities associated with each teleparallel geometry. In this paper, we will propose a modification of the Cartan-Karlhede algorithm which is applicable to {\it any} teleparallel geometry, regardless of the field equations. Using the algorithm we will show that any theory of gravity admitting non-zero torsion will not permit a maximally symmetric space except Minkowski space. Furthermore this result implies that the group of affine frame symmetries for any $4$-dimensional (4D) Riemann-Cartan geometry with non-zero torsion is at most seven-dimensional.

The outline of the paper is as follows. In section \ref{sec:BasicTEGR}, we briefly review the formulation of covariant $f(T)$ teleparallel gravity. In section \ref{sec:CKalg} we introduce the Cartan-Karlhede (CK) algorithm adapted to teleparallel geometries, discuss the concept of an affine frame symmetry and use affine frame symmetries to determine the symmetry group of a teleparallel geometry. In section \ref{sec:TorsionDecomp}, we classify the irreducible parts of the torsion tensor using the alignment classification and identify the linear isotropy group of the torsion tensor. We also present a rank four tensor whose trace appears in the field equations of TEGR. Using this classification we determine the largest group of symmetries for a Riemann-Cartan space admitting non-zero torsion and the highest iteration of the Cartan-Karlhede algorithm needed to classify teleparallel geometries.  In section \ref{sec: Examples} we illustrate the algorithm by applying it to six well-known examples. In section \ref{sec:discussion} we summarize our results and discuss future work.

\subsection{Notation}

As in GR, the formulation of teleparallel theories of gravity requires similar notation. We will denote the coordinate indices by $\mu, \nu, \ldots$ and the tangent space indices by $a,b,\ldots$. Unless otherwise indicated the spacetime coordinates will be $x^\mu$. The frame fields are denoted as $\bh_a$ and the dual coframe one-forms are $\bh^a$. The vielbein components are $h_a^{~\mu}$ or $h^a_{~\mu}$. For a given anholonomic coframe, $\bh^a$, the structure coefficients or coefficients of anholonomy are $C^c_{~ab}$. The spacetime metric will be denoted as $g_{\mu \nu}$ while the Minkowski tangent space metric is $\eta_{ab}$. We will denote the fully anti-symmetric Levi-Civita symbol as $\epsilon_{abcd}$ with $\epsilon_{1234} = 1$. 

To denote a local Lorentz transformation leaving $\eta_{ab}$ unchanged, we write $\Lambda_a^{~b}(x^\mu)$. The spin-connection one-form $\bomega^a_{~b}$, is designated by $\bomega^a_{~b} = \omega^a_{~bc} \bh^c$. The curvature and torsion tensors will be denoted, respectively, as $R^a_{~bcd}$ and $T^a_{~bc}$, in addition we will work with the torsion two-form, ${\bf T}^a = \frac12 T^a_{~bc} \bh^b \wedge \bh^c$. The irreducible parts of the torsion tensor: the vector part, the axial part and the purely tensorial part will be denoted as $V^a, A^a$ and $t_{(ab)c}$, respectively.

Covariant derivatives with respect to a Levi-Civita connection will be denoted using a semi-colon, $T_{abc;e_1}$. A vertical bar will denote covariant differentiation with respect to the spin-connection, $T_{abc|e_1}$ or as $D_{e_1} T_{abc}$. We will also write $D_{\mu } = h^a_{~\mu} D_a$. 


We will denote the set of components of the torsion tensor and its covariant derivatives up to $q$-th order as $\mathcal{T}^q = \{ t_{abc}, t_{abc|e_1}, \ldots _{abc|e_1 \ldots e_q} \}$, in addition we write $\mathcal{T} = \mathcal{T}^p$ for the entire set of Cartan invariants needed to classify a teleparallel geometry, where the final iteration is $q=p$. The number of functionally independent components of the torsion tensor and its covariant derivatives up to order $q$ will be recorded by $t_q$. The linear isotropy group of the $q$-th covariant derivative of the torsion tensor will be denoted as $H_q$.

\newpage 
\section{A brief review of teleparallel gravity}\label{sec:BasicTEGR}

As with  GR and its modifications, there are many possible generalized teleparallel gravity theories \cite{Krssak:2018ywd}, but we will focus on the well-studied $f(T)$ theories. Let $M$ be a 4D differentiable manifold $M$ with coordinates $x^\mu$.  Due to the differentiability of $M$, there exists a non-degenerate coframe field $\bh^a$ which is defined on a subset $U\subset M$.  To achieve a notion of geometry (lengths and angles) we also assume the existence of a symmetric metric field on $M$, with coordinates $g_{ab}$ with respect to this coframe field.  Further, in order to compare quantities at different points or equivalently to define a covariant differentiation process, we require a notion of parallel transport, and assume the existence of a linear affine connection one form $\bomega^{a}_{\phantom{a}b}$.  In general, the geometrical quantities $g_{ab},\bh^a,\bomega^{a}_{\phantom{a}b}$ are independent and have 10, 16, and 64 independent elements respectively.  It is only when we begin to make additional assumptions that various constraints and relationships between $g_{ab},\bh^a,\bomega^{a}_{\phantom{a}b}$ become apparent.

\subsection{Gauge choices}

We will assume the ({\em Principle of Relativity}) and the consequent covariance of the field equations under local $GL(4,\mathbb{R})$ gauge transformations. We take advantage of this gauge freedom of local linear transformations in the tangent space to judiciously simplify aspects of the calculations. The {\bf Orthonormal gauge} is a gauge choice that diagonalizes the tangent space metric $g_{ab}= \eta_{ab} = \mbox{Diag}[-1,1,1,1]$ and since $g_{ab}$ is symmetric this is always possible. A useful alternative is the {\bf Null Gauge} in which the tangent space metric becomes, $$g_{ab} = \left[ \begin{array}{cccc} 0 & -1 & 0 & 0 \\ -1 & 0 & 0 & 0 \\ 0 & 0 & 1 & 0 \\ 0 & 0 & 0 & 1 \end{array}\right]. $$ Further, one can also use the {\bf Complex Null Gauge} in which case, $$g_{ab} = \left[ \begin{array}{cccc} 0 & -1 & 0 & 0 \\ -1 & 0 & 0 & 0 \\ 0 & 0 & 0 & 1 \\ 0 & 0 & 1 & 0 \end{array}\right]. $$

There still exists a $O(1,3)$ subgroup of $GL(4,\mathbb{R})$ of residual gauge transformations  of the coframe that leaves the metric $g_{ab}=\eta_{ab}$ invariant.  Since one typically also wants to preserve the orientation of space and the direction of time, one further restricts this subgroup of allowable transformations to be the proper ortho-chronous Lorentz subgroup, $SO(1,3)$ of $O(1,3)$. We shall denote these Lorentz transformations of the coframe as: $$ \bh^a \to \Lambda^a_{~b} {\bh'}^b. $$ 

Hence, if one chooses the orthonormal gauge, then the resulting field equations must transform homogeneously under any remaining gauge freedom, in this case $SO(1,3)$ Lorentz transformations. Since the metric is completely fixed in this gauge, only the coframe and spin-connection are independent dynamical variables yielding 16+64=80 independent elements albeit with a remaining 6 dimensional $SO(1,3)$ gauge freedom.

With the additional assumption that the spin-connection be metric compatible, i.e., $Q_{abc}\equiv  - D_{c} g_{ab} =0$, the spin-connection becomes anti-symmetric, $\bomega_{(ab)}=0$.  Due to the algebraic nature of this constraint, it can be implemented easily without loss of generality. The fundamental variables remaining are the 16 elements of the coframe $\bh^a_{\phantom{a}}$ and the 24 elements of the anti-symmetric spin-connection $\bomega^a_{\phantom{a}b}$.  The torsion and the curvature associated with the coframe and spin-connection are
\begin{eqnarray}
T^a_{\phantom{a}\mu\nu}&=&\partial_\mu h^a_{\phantom{a}\nu}-\partial_\nu h^a_{\phantom{a}\mu}+\omega^a_{\phantom{a}b\mu}h^b_{\phantom{a}\nu}-\omega^a_{\phantom{a}b\nu}h^b_{\phantom{a}\mu},\\
R^a_{\phantom{a}b\mu\nu} &=& \partial_\mu \omega^a_{\phantom{a}b\nu}-\partial_\nu \omega^a_{\phantom{a}b\mu}+\omega^a_{\phantom{a}c\mu}\omega^c_{\phantom{a}b\nu}-\omega^a_{\phantom{a}c\nu}\omega^c_{\phantom{a}b\mu}.
\end{eqnarray}

\subsection{The torsion scalar and the gravitational Lagrangian}

To derive the field equations for teleparallel gravity, we will consider a Lagrangian for $f(T)$ teleparallel gravity where the scalar quantity, $T$, is defined as
\begin{equation}
T =
\frac{1}{4} \; T^\rho_{\phantom{\rho}\mu\nu} \, T_\rho^{\phantom{\rho}\mu \nu} +
\frac{1}{2} \; T^\rho_{\phantom{\rho}\mu\nu} \, T^{\nu \mu}_{\phantom{\rho\rho}\rho} -
            T^\rho_{\phantom{\rho}\mu\rho} \, T^{\nu \mu}_{\phantom{\rho\rho}\nu}.
\label{TeleLagra}
\end{equation}
\noindent The scalar $T$ can be written more compactly using the super-potential, $S^a_{~\mu \nu}$, expressed in terms of the torsion tensor
\begin{equation}
S_a^{\phantom{a}\mu\nu}=\frac{1}{2}\left(T_a^{\phantom{a}\mu\nu}+T^{\nu\mu}_{\phantom{\nu\mu}a}-T^{\mu\nu}_{\phantom{\mu\nu}a}\right)-h_a^{\phantom{a}\nu}T^{\rho \mu}_{~~\rho} + h_a^{\phantom{a}\mu}T^{\rho \nu}_{~~\rho},
\end{equation} 
\noindent  so that:
\begin{equation}
T=\frac{1}{2}T^a_{\phantom{a}\mu\nu}S_a^{\phantom{a}\mu\nu}.
\end{equation}

The structure of the gravitational Lagrangian for $f(T)$ teleparallel gravity theories is then
\begin{equation}
L_{Grav}(h^a_{\phantom{a}\mu},\omega^a_{\phantom{a}b\mu})=\frac{h}{2\kappa}f(T).
\end{equation}
\noindent where $h$ is the determinant of the vielbein matrix $h^a_{~\mu}$. Since we must account for the non-trivial transformational properties of the spin-connection, we compute all variations of the action and  include a non-trivial spin-connection \cite{Krssak_Saridakis2015}. 

\subsection{Matter Lagrangian and the canonical energy momentum}

To describe non-vacuum teleparallel gravity solutions, we will include a matter Lagrangian to supplement the teleparallel Lagrangian. Assuming that the matter field $\Phi$ is spinless and minimally coupled to gravity, the matter Lagrangian $L_{Matt}$ will be a function of the tetrad field $h^a_{\phantom{a}\mu}$, the matter field $\Phi$ and its covariant derivatives, $D_\nu\Phi$,
\begin{equation}
L_{Matt}=L_{Matt}(h^a_{\phantom{a}\mu},\Phi,D_\nu\Phi).
\end{equation}
We define the canonical energy momentum as,
\begin{equation}
h\Theta_a^{\phantom{a}\mu}=-\frac{\delta L_{Matt}}{\delta h^{\vphantom{\mu}a}_{\phantom{a}\mu}}.
\end{equation}

\subsection{Lagrangian formulation of $f(T)$ teleparallel gravity}

The field equations are calculated by varying the action:
\beq S=\int d^4x\Bigl(L_{Grav}(h^a_{\phantom{a}\mu},\omega^a_{\phantom{a}b\mu})+L_{Matt}(h^a_{\phantom{a}\mu},\Phi,D_\nu\Phi)\Bigr). \label{TPaction} \eeq
\noindent If we consider the spin-connection as an independent quantity having zero curvature, then the gravitational Lagrangian is of the form:
\begin{equation}
L_{Grav}(h^a_{\phantom{a}\mu},\omega^a_{\phantom{a}b\mu},\nu_a^{\phantom{a}b\mu\nu})=\frac{h}{2\kappa}f(T)+\nu_a^{\phantom{a}b\mu\nu}R^a_{\phantom{a}b\mu\nu},
\end{equation}
the Lagrange multiplier has been introduced to address the condition of the connection with zero curvature, which is anti-symmetric in the last two indices
\begin{equation}
\nu_a^{\phantom{a}b(\mu\nu)}=0,
\end{equation}
The variation of \eqref{TPaction} with respect to $\Phi$, $\nu_a^{\phantom{a}b\mu\nu}$, $h^a_{\phantom{a}\mu}$, and $\omega^a_{\phantom{a}b\mu}$ yield:
\begin{eqnarray}
0 &=& \frac{\delta L_{Matt}(h^a_{\phantom{a}\mu},\Phi,D_\nu\Phi)}{\delta \Phi},\\
0&=& \partial_\mu\omega^a_{\phantom{a}b\nu}-\partial_\nu\omega^a_{\phantom{a}b\mu}+\omega^a_{\phantom{a}e\mu}\omega^e_{\phantom{a}b\nu}-\omega^a_{\phantom{a}e\nu}\omega^e_{\phantom{a}b\mu},\label{zero-curvature}\\
\kappa \Theta_a^{\phantom{a}\mu}&=&
        h^{-1}f_T(T)\partial_v\left(hS_a^{\phantom{a}\mu\nu}\right)+f_{TT}(T)S_a^{\phantom{a}\mu\nu} \partial_v T \nonumber\\
        &&+\frac{1}{2}f(T)h_a^{\phantom{a}\mu} -f_T(T)T^b_{\phantom{a}a\nu}S_b^{\phantom{a}\mu\nu}-f_T(T)\omega^b_{\phantom{a}a\nu}S_b^{\phantom{a}\mu\nu},\label{EQ2b}\\
0&=&
        2\kappa h^{-1}( D_{\nu} \nu_a^{~b\mu \nu} ) +f_T(T)h^b_{\phantom{a}\nu} S_a^{\phantom{a}\mu\nu}.\label{hypermomentum} \end{eqnarray}
These equations determine the spin-connection, the coframe and the Lagrange multipliers, up to some gauge transformations. We shall regard the $f(T)$ theory as being defined by the resulting field equations which are {\it Lorentz covariant}.

\subsubsection{Interpretation of the field equations}

Equations \eqref{hypermomentum} are an underdetermined set of equations for the Lagrange multipliers, and hence determines a subset of these. These equations are not required to determine the actual dynamics of the coframe. It is equation \eqref{EQ2b} which yields the field equations for determining the coframe via $E_a^{~b} = \kappa h \Theta_a^{~b}$, where we define the right-hand side of equation \eqref{EQ2b} as $E_a^{\phantom{a}\mu}$ and define $E_{ab}=g_{bc}h^c_{\phantom{a}\mu}E_a^{\phantom{a}\mu}$

The vanishing of the non-metricity $Q_{abc} = -D_{c} g_{ab} =0 $ imposes the constraint that the connection one-form is anti-symmetric, $\bomega_{ab} = - \bomega_{ba}$; we will implicitly assume this hereafter.  The condition that $R^a_{~bcd}=0$ imposes a second differential constraint on the spin-connection and implies that the most general solution to equation \eqref{zero-curvature} is of the form: \begin{equation}
\omega^a_{\phantom{a}b\mu} = \Lambda^a_{\phantom{a}c}\partial_\mu\Lambda_{b}^{\phantom{a}{c}}\label{solution_omega}
\end{equation}
where $\Lambda^a_{\phantom{a}b}$ is some yet undetermined Lorentz transformation. 

There is a third constraint on $\bomega^a_{~b}$ arising from \eqref{EQ2b}, in addition to determining the coframe. Imposing invariance under infinitesimal Lorentz transformations,  the matter Lagrangian now requires that $$\Theta_{[ab]} = 0,$$ so that the canonical energy momentum is symmetric. Furthermore, this connects the metrical energy momentum $T_{ab}$ and the canonical energy momentum: $$T_{ab} = - \frac12 \frac{\delta L_{Matt}}{\delta g_{ab} } = \Theta_{(ab)}. $$

We can separate equations \eqref{EQ2b} into it symmetric and anti-symmetric parts. This symmetric, anti-symmetric and matter field equations describing $f(T)$ gravity are
\begin{eqnarray}
\kappa \Theta_{(ab)}&=&
        f_{TT}(T)S_{(ab)}^{\phantom{(ab)}\nu} \partial_v T+f_T(T)\GLC_{ab} + \frac{1}{2}g_{ab}\left(f(T)-Tf_T(T)\right),\label{temp1}\\
             0      &=& f_{TT}(T)S_{[ab]}^{\phantom{[ab]}\nu} \partial_v T,\label{temp2}\\
             0      &=& \frac{\delta L_{Matt}(h^a_{\phantom{a}\mu},\Phi,D_\nu\Phi)}{\delta \Phi}.\label{temp3}
\end{eqnarray}

\noindent where $\GLC_{ab} = \GLC_{\mu \nu} h_a^{~\mu} h_b^{~\nu}$ where $\GLC_{\mu \nu}$ is the Einstein tensor associated with the spacetime metric $g_{\mu \nu} = \eta_{ab} h^a_{~\mu} h^b_{~\nu}$.

For $f(T)_{} \neq T$, it can be shown that the variation of the gravitational Lagrangian by the flat spin-connection is equivalent to the anti-symmetric part of the field equations in equation \eqref{temp2}. This provides the third condition on $\omega^a_{~b}$ and thus if the field equations for the spin-connection are satisfied then the field equations for the coframe are guaranteed to be symmetric \cite{Hohmann:2018rwf}. This motivates the search for an appropriate connection.

\subsection{The choice of the spin-connection}


The selection of the spin-connection within the covariant formulation of the $f(T)$ theory complicates the search for solutions to the field equations. In the case of TEGR, where $f(T)=T$, equation \eqref{temp2} vanishes.  For TEGR, any initial ansatz for the coframe $h^a_{\phantom{a}\mu}$ will contain effects due to both inertia and gravitation.  Separating the gravitational effects from the inertial effects, for a given $h^a_{\phantom{a}\mu}$ we will suppose $e^a_{\phantom{a}\mu}$ denotes the frame where only the inertial effects are present and the gravitational effects are zero \cite{Aldrovandi_Pereira2013}.  We can then use this inertial frame to determine the spin-connection $\omega^a_{\phantom{a}b\mu}$ which encodes the inertial effects present in the coframe, since the torsion and curvature of the inertial coframe and associated spin-connection satisfy
\begin{equation}
T^a_{\phantom{a}\mu\nu}[e^a_{\phantom{a}\rho},\omega^a_{\phantom{a}b\mu}] = 0, \qquad R^a_{\phantom{a}b\mu\nu}[e^a_{\phantom{a}\rho}, \omega^a_{\phantom{a}b\mu}]=0.
\end{equation}
So for the inertial coframe, $e^a_{\phantom{a}\mu}$, we must determine the conditions under which this yields a trivial torsion tensor. The torsion tensor of an inertial coframe can be written in terms of the coefficients of anholonomy, $C^a_{~bc}$, arising from the Lie bracket of the frame basis $e^a_{~\mu}$ and the spin-connection, $\omega^a_{~bc}$:
\begin{equation}
T^a_{\phantom{a}bc}[e^a_{\phantom{a}\mu},\omega^a_{\phantom{a}b\mu}]=C^a_{\phantom{a}bc}[e^a_{\phantom{a}\mu}]+\omega^a_{\phantom{a}cb}-\omega^a_{\phantom{a}bc} = 0,
\end{equation}
We can use different combinations of indices to solve for $\omega^a_{\phantom{a}b\mu}$:
\begin{equation}
\omega^a_{\phantom{a}b\mu}=\frac{1}{2}e^c_{\phantom{a}\mu}\Bigl(C^{\phantom{a}a}_{b\phantom{a}c}[e^a_{\phantom{a}\mu}] + C^{\phantom{a}a}_{c\phantom{a}b}[e^a_{\phantom{a}\mu}] - C^a_{\phantom{a}bc}[e^a_{\phantom{a}\mu}]\Bigr).\label{conn1}
\end{equation}

The above spin-connection is always a valid connection to use for TEGR and can be used for more general $f(T)$ theories \cite{Krssak2017}. However, we still have the challenge of selecting the frame $e^a_{\phantom{a}\mu}$ which contains only the inertial effects.

\subsection{The covariant form of the $f(T)$ field equations}

Starting from equation \eqref{temp1}, the $f(T)$ field equations can be restated purely in a frame basis: 
\beq D_c (f_{T} S_a^{~bc}) + f_{T}[S_a^{~bc} T^d_{~cd} + \frac12 S_a^{~cd} T^b_{~cd} + S_d^{~bc} T^d_{~ca}] + \frac12 \delta_a^{~b} f(T) =  \kappa  \Theta_a^{~b} \eeq



\noindent We can define an associated rank four tensor as

\beq F_{ac}^{~~bd} =  D_c (f_{T} S_a^{~bd}) + f_{T}[S_a^{~bd} T^e_{~ce} - \frac12 S_a^{~e[b} T^{d]}_{~ce} + S_e^{~bd} T^e_{~ca}] + \frac16 \delta_a^{~[b} \delta_c^{~d]} f(T) \eeq



\noindent the field equations now occur as a contraction of $F_{ab}^{~~cd}$. We can classify $F_{ab}^{~~bd} \equiv F_a^{~d} = \kappa h \Theta_a^{~d}$ to distinguish between $f(T)$ solutions. 

As $F_{ac}^{~~bd}$ satisfies $F_{ab}^{~~(cd)} = 0$ there are only three traces:
\beq E_a^{~b} = F_{ac}^{~~bc},~~H_a^{~b} = F_{ca}^{~~bc} \text{ and } I^{ab} = F_{cd}^{~~ab} g^{cd} = F^{c~ab}_{~c}. \eeq
\noindent Furthermore this rank four tensor has only one complete  trace since
\beq E_a^{~a} = E,~~H_a^{~a} = -E_a^{~a} = -E \text{ and } I^a_{~a} = 0. \eeq
\noindent In a similar manner, we can classify the irreducible parts of $F_{ad}^{~~bc}$ itself and this potentially can give insight into the physical interpretation of $f(T)$ solutions. For example, since the trace, $\Theta_a^{~b} =  (\kappa h )^{-1} F_{a}^{~b}$, is involved in the field equations and is associated with the energy-momentum tensor, one could consider the completely trace-free object with all indices lowered:
\beq \begin{aligned} \bar{C}_{abcd} & \equiv F_{abcd} + 2 X_{a[c} g_{d]b} + 2 Y_{b[c}g_{d]a} + g_{ab} Z_{[cd]} - \frac{E}{12} (g_{ac} g_{bd} - g_{ad} g_{bc}), \end{aligned} \eeq

\noindent where 
\beq \begin{aligned} X_{ab} &= -\frac38 E_{ab} + \frac18 H_{ab} - \frac{1}{24} E_{[ab]} - \frac{1}{24} H_{[ab]} - \frac{1}{12} I_{ab} + \frac{E}{8} g_{ab}, \\
Y_{ab} &= -\frac38 H_{ab} + \frac18 E_{ab} - \frac{1}{24} E_{[ab]} - \frac{1}{24} H_{[ab]} - \frac{1}{12} I_{ab} - \frac{E}{8} g_{ab} \\
Z_{[ab]} &= - \frac16 E_{[ab]} - \frac16 H_{[ab]} - \frac{1}{3} I_{ab}. \end{aligned} \eeq

\noindent This tensor can be treated as an analogue of the Weyl tensor (it is of interest to determine if this describes a similar aspect of the gravitational field). The manifold can now be classified via the algebraic classification (e.g., alignment classification) of $\Theta_a^{~b}$ and $\bar{C}_{ab}^{~~cd}$, in analogy with the Ricci and Weyl tensor in the Riemannian case.



\newpage
\section{The equivalence problem for teleparallel geometries} \label{sec:CKalg}

To determine the equivalence of two teleparallel geometries,, we must consider the Lorentz frame bundle of orthonormal frames, $\mathcal{F}(M)$. We will raise the coframe basis, $\{ \bh^a \}$ on $M$ to $\mathcal{F}(M)$ and denote it as $\{ {\bf \tilde{h}}^a \}$. Then, in a similar manner to the coframe, we can raise the spin-connection $\omega^a_{~bc}$ to act as a connection on the frame manifold, $\tilde{\omega}^a_{~bc}$, and define {\it the connection one-forms} as: 
$$\tilde{\bomega}^a_{~b} = \tilde{\omega}^a_{~bc} {\bf \tilde{h}}^c. $$ 
For $\mathcal{F}(M)$ there exists a uniquely defined basis of the cotangent space of $\mathcal{F}(M)$, $\{{\bf \tilde{h}}^a, \tilde{\bomega}^a_{~b}\}$, since the cotangent space can be decomposed into the direct sum of the 4D horizontal subspace, spanned by the canonical one-forms $\{ {\bf \tilde{h}}^a \}$, and the 6-dimensional vertical subspace, spanned by the linear connection one-forms $\{ \tilde{\bomega}^a_{~b}\}$ defined on $\mathcal{F}(M)$ \cite{Kobayashi_Nomizu1996}. The definition of equivalence can be restated in the framework of the frame bundle manifold: $M$ and $M'$ are the same if and only if their corresponding frame bundles are the same \cite{fon1992, fon1996, kramer}. 

Thus, $M$ and ${M}'$ are locally equivalent when there is a diffeomorphism, $F$, from $\mathcal{F}(M)$ onto ${\mathcal{F}'}(M')$ such that the pullback satisfies: \beq F^* {{\bf \tilde{h}}}^{'a} = {\bf \tilde{h}}^a \text{ and } F^*{\tilde{\bomega}}^{'a}_{~~b} = \tilde{\bomega}^a_{~b}. \label{Fmap} \eeq   
\noindent For metric-based theories the first condition is sufficient to determine the equivalence of $M$ and ${M}'$ as the connection is computed in terms of $h^a_{~\mu}$ and their coordinate derivatives. However, for theories where the connection is not metric compatible or symmetric, the second condition is necessary to ensure that the connection is invariant under a pullback as well. 

With this new manifold, the solution to the equivalence problem can be obtained using Cartan's result on the equivalence of one-forms with regard to the set $\{ {\bf \tilde{h}}^a, \tilde{\bomega}^a_{~b} \}$. However, we would like to work with quantities on $M$, we will use Cartan's structure equations for a general connection:
\beq \begin{aligned}  d {\bf \tilde{h}}^a &= -\tilde{\bomega}^a_{~b} \wedge {\bf \tilde{h}}^b + \frac12 T^a_{~bc} {\bf \tilde{h}}^b \wedge {\bf \tilde{h}}^c, \\
d\tilde{\bomega}^a_{~b} &= -\tilde{\bomega}^a_{~c} \wedge \tilde{\bomega}^c_{~b} + \frac12 R^a_{~bcd} {\bf \tilde{h}}^c \wedge {\bf \tilde{h}}^d, \end{aligned} \label{CartanStructure} \eeq

\noindent where $T^a_{~bc}$ and $R^a_{~bcd}$ denote the torsion tensor and curvature tensor, respectively, to determine scalars that can be used on $M$. For those connections where $\tilde{\bomega}^a_{~b}$ leads to vanishing curvature, this shows that the torsion tensor can be used to compare the two frame bundles. Applying the exterior derivative to the components of the torsion two-form $T^a = \frac12 T^a_{~bc} {\bf \tilde{h}}^b \wedge {\bf \tilde{h}}^c$ on $\mathcal{F}(M)$ lead to quantities involving the covariant derivative of the torsion tensor:
\beq d T^a_{~bc} = T^a_{~bc|d} \bh^d - T^d_{~bc} \tilde{\bomega}^a_{~d} + T^a_{~dc} \tilde{\bomega}^d_{~b} + T^a_{~bd} \tilde{\bomega}^d_{~c}. \label{dtorsion} \eeq

Solving for $T^a_{~bc|d}$ and applying the exterior derivative to the components yields expressions in terms of the second order covariant derivatives of the torsion tensor. This procedure can be repeated for higher order derivatives of the torsion tensor. Therefore, we can determine the equivalence of $M$ and ${M}'$ by considering the covariant derivatives of the torsion tensor when viewed as functions on $\mathcal{F}(M)$ and $\mathcal{F}'({M}')$, respectively. We note that for a general Riemann-Cartan geometry, $R_{abcd}$, $T_{abc}$ and their covariant derivatives can always be isolated using the identities in \eqref{CartanStructure} \cite{fon1992}.

We will denote the components of the torsion tensor and its covariant derivatives up to order $q$ for a manifold $M$ relative to a chosen orthonormal frame on $M$ as $\mathcal{T}^q$. If there are $k$ elements in a maximal set of functionally independent invariants on $\mathcal{F}(M)$, we will index such invariants as $I^\alpha,~~\alpha = 1, \hdots, k$ and we will call the set of indices of the corresponding $T^a_{~bc}$ their {\it index basis} and denote this set as $\mathcal{A}$. 

Employing the frame bundle perspective, by denoting $p$ as the last derivative at which a new functionally independent quantity arises in $T^p$ we have the following theorem for comparing teleparallel geometries when $n=4$ \cite{fon1992}: 

\begin{thm}
Let $M$ and ${M}'$ be spacetimes of differentiability $C^{12}$, $x$ be a regular point of $M$, $\bh^a $ be a frame at $x^\mu$ and $\bomega^a_{~b}$ be a spin-connection at $x^\mu$, and similarly ${\bomega}^{'a}_{~~b}$ at ${x'}^\mu$ for ${M}'$. There is a diffeomorphism which maps $(x^\mu,\bh^a , \bomega^a_{~b})$ to $({x'}^\mu, {\bh'}^a, {\bomega}^{'a}_{~~b})$ if and only if $\mathcal{T}^{p+1}$ for $M$ is such that 
\begin{enumerate}
\item $\mathcal{A}$ indexes ${I}^{'\alpha}$ which are functionally independent in $\mathcal{F}'({M}')$, 
\item $I^\alpha(x^\mu,\bh,\bomega^a_{~b}) = {I}^{'\alpha}({x'}^\mu, {\bh'}, {\bomega}^{'a}_{~~b})$ for $\alpha = 1, \hdots, k$ and
\item All other components of $\mathcal{T}^{p+1}$ and $\mathcal{T}^{'p+1}$ expressed in terms of the $I^\alpha$ and ${I}^{'\alpha}$, respectively, are the same for $M$ and ${M}'$.
\end{enumerate}
\end{thm}


We will consider a specialized version of the algorithm proposed for Riemann-Cartan spacetimes \cite{fon1992} which has been implemented in CLASSI \cite{fon1996}. This has been previously applied to non-covariant teleparallel gravity \cite{fon2002}. The pure tetrad teleparallel gravity condition leads to a significantly restricted equivalence algorithm for these torsion theories as the only permitted Lorentz transformations are those with constant parameters, the so-called {\it global Lorentz transformations}. 

\subsection{ The modified Cartan-Karlhede algorithm for torsion}
Again, using $\mathcal{T}^q$ to denote the set of components of the torsion tensor and its the derivatives of the torsion tensor up to the $q$th order, i.e., $\mathcal{T}^q \equiv \{ T_{abc}, T_{abc|d_1}, \ldots T_{abc|e1 \ldots e_q} \}$, then for any teleparallel geometry the Cartan-Karlhede algorithm is then:

\begin{enumerate}
\item Set the order of differentiation $q$ to 0.
\item Calculate $\mathcal{T}^q$.
\item Determine the canonical form of the $q$-th covariant derivative of the torsion tensor.
\item Fix the frame as much as possible, using this canonical form, and record the remaining frame transformations that preserve this canonical form (the group of allowed frame transformations is the {\it linear isotropy group $H_q$}). The dimension of $H_q$ is the dimension of the remaining {\it vertical} freedom of the frame bundle.
\item Find the number $t_q$ of independent functions of spacetime position in $\mathcal{T}^q$ in the canonical form. This tells us the remaining {\it horizontal} freedom.
\item If the dimension of $H_q$ and number of independent functions are the same as in the previous step, let $p+1=q$, and the algorithm terminates; if they differ (or if $q=0$), increase $q$ by 1 and go to step 2. 
\end{enumerate}

\noindent The resulting non-zero components of $\mathcal{T}^p$ constitute the {\it Cartan invariants} and we will denote them as $\mathcal{T} \equiv \mathcal{T}^{p+1}$ so that 
\beq \mathcal{T} = \{ T_{abc}, T_{abc|d_1}, \ldots T_{abc|d_1 \ldots d_{p+1}} \}. \label{CartanSet} \eeq 
We will refer to the invariants constructed from, or equal to, Cartan invariants of any order as {\it extended Cartan invariants}. We have not specified the form of the invariant frame used in the algorithm. We will discuss how {\it the frame can be fixed invariantly at zeroth and first order in the next section}. 

For sufficiently smooth frames and spin-connection, the result of the algorithm is a set of scalars providing a unique local geometric characterization of the teleparallel geometry. The 4D spacetime is characterized by the canonical form used, the two discrete sequences arising from the successive linear isotropy groups and the independent function counts, and the values of the (non-zero) Cartan invariants. As there are $t_p$ essential spacetime coordinates, the remaining $4-t_p$ are ignorable, and so the dimension of the affine frame symmetry isotropy group (hereafter called the isotropy group) of the spacetime will be $s=\dim(H_p)$ and the affine frame symmetry group has dimension: \beq r=s+4-t_p. \label{rnumber} \eeq

\subsection{Scalar polynomial torsion  invariants}
As in metric based theories, there is an alternative set of invariants that are available from the torsion tensor and its covariant derivatives. These invariants are called the {\it scalar polynomial (torsion) invariants} (SPIs) and are constructed from full contractions of tensors built out of copies of the torsion tensor and its covariant derivatives with respect to the spin-connection. An example of an SPI is the torsion scalar $T$. We will denote the set of such SPIs as $\mathcal{I}_T$. While this set is functionally infinite, there is a finite dimensional basis of SPIs. It is a question of interest as to how the set $\mathcal{I}_T$ is related to the set of SPIs formed from the curvature tensor and its covariant derivatives, $\mathcal{I}_R$ \cite{CSI4a}. 

As a generalization of $f(T)$ theories, we can consider new gravitational Lagrangians constructed by including other SPIs: 
\beq \mathcal{L}_{Grav} (h^a_{~\mu}, \omega^a_{~b \mu} ) = \frac{h}{2	\kappa} f( T, T_1, T_2, \ldots T_n), T_i \in \mathcal{I}_T,~ i \in [1,n]. \eeq 
\noindent By including a matter Lagrangian and varying the resulting action, we can generate new teleparallel gravity theories that are distinct from the $f(T)$ theories.

A natural question to ask is whether the SPIs are able to uniquely characterize teleparallel geometries, or if there are cases where the SPIs are unable to distinguish between teleparallel geometries. For example, it is possible that there are teleparallel geometries where all SPIs vanish, and hence have the $\text{VSI}$ property \cite{{Higher}}. For such $\text{VSI}$ torsion geometries, it would be impossible to distinguish them from Minkowski space with SPIs. A related question is when Cartan invariants are uniquely related to SPIs \cite{Coley:2009eb} and when the frame bundle can be completely characterized by SPIs. See the PPGW example later in the paper.

\subsection{Affine frame symmetries} \label{subsec:FSs}

\subsubsection{Affine frame symmetries in general relativity}

To clarify the formula for the dimension of group of affine frame symmetries, we will briefly review a well-known result of Killing vectors in GR. We will also assume that the isotropy group of $M$ is trivial. If a vector field ${\bf X}$ satisfies 
\beq \mathcal{L}_{{\bf X}} g_{ab} = 0, \label{def:iso} \eeq
then since the Levi-Civita connection is expressed in terms of the metric and its coordinate derivatives, the ``Lie derivative of the connection'' will vanish automatically, implying that they are affine collineations as well \cite{fon1992}. The Lie derivative of the curvature tensor and its covariant derivatives with respect to ${\bf X}$ will also vanish:
\beq \begin{aligned} \mathcal{L}_X R_{abcd} &= 0,  \\ \mathcal{L}_X R_{abcd;e_1\dots e_q} &= 0,~q\in[1,p].\end{aligned} \label{RiemSym} \eeq

If instead we assume that equation \eqref{RiemSym} is satisfied for a vector field {\bf X}, then relative to the invariant coframe determined by the Cartan-Karlhede algorithm, $\bh^a$, a set of linearly independent vector fields satisfy \eqref{RiemSym}. It follows that any invariant, $I \in \mathcal{R}$ must be annihilated under frame differentiation, i.e., ${\bf X}[I] = 0$. This implies that there exists a set of Killing vector fields, $\mathcal{K}$ for which  ${\bf X}' \in \mathcal{K}$ coincides with ${\bf X}$ at any point \cite{Hervik2011, Hervik:2010rg}. Furthermore, if \eqref{RiemSym} holds for a given vector field ${\bf X}$, then using rectifying coordinates for ${\bf X}$ requires that \beq  \mathcal{L}_{{\bf X}} \bh^a = 0. \label{LDframe} \eeq As the Levi-Civita connection is expressed in terms of $h^a_{~\mu}$ and their derivatives, it follows that the Levi-Civita connection relative to the frame basis will be annihilated under Lie differentiation so that \eqref{Intro:FS2} is satisfied and the vector field ${\bf X}$ is an affine frame symmetry.

\subsubsection{Affine frames symmetries in teleparallel geometries}

To determine the equivalence of the respective frame bundles of the two teleparallel geometries, we will compare the set $\mathcal{T}$ in equation \eqref{CartanSet} for each. If we are interested in diffeomorphisms, $\Phi$, from the manifold to itself that preserves the form of the metric $g = \eta_{ab} \bh^a \bh^b$ along with the associated structure of the frame bundle $\mathcal{F}(M)$, then we must consider the mappings acting on the uniquely defined linearly independent basis of the cotangent space of $\mathcal{F}(M)$, $\{ {\bf \tilde{h}}^a, \tilde{\bomega}^a_{~b} \}$ so that 
\beq \Phi^* {\bf \tilde{h}}^a = {\bf \tilde{h}}^a \text{ and } \Phi^* \tilde{\bomega}^a_{~b} = \tilde{\bomega}^a_{~b}. \eeq 
Equivalently, we can determine the set of infinitesimal generators, ${\bf X}$,
\beq  \mathcal{L}_{{\bf X}} {\bf \tilde{h}}^a = \mathcal{L}_{{\bf X}} \tilde{\bomega}^a_{~b} = {\bf 0}. \eeq
\noindent As the torsion tensor and its covariant derivatives are expressed in terms of the frame and the connection then the Lie derivative of ${\bf X}$ satisfies:
\beq \mathcal{L}_{{\bf X}} T_{abc|e_1 \ldots e_q}=0,~~q \in [1,p]\eeq
\noindent  which motivates the definition of an affine frame symmetry.

\subsection{Symmetries in flat geometries with torsion} \label{subsec:TPSymmetries}
We have generated a set of invariants that characterize a teleparallel geometry, determined the isotropy group and have fixed the frame as much as possible using the torsion tensor and its covariant derivatives. From the number of functionally independent invariants we can determine the number of diffeomorphisms of the manifold which preserves the form of the torsion tensor and its covariant derivatives in analogy with the isometries of a metric in GR which preserves the form of the Riemann tensor and its covariant derivatives. 

When the equivalence algorithm results in a frame where all frame freedom is fixed, the linear isotropy group is trivial and the coframe is called an {\it invariant coframe}. An invariant coframe allows us to determine the affine frame symmetries. For example, if $\Phi$ is an affine frame symmetry and $\{ \bh^a\}$ is an invariant coframe then, 
\beq \Phi^{*} \bh^a = \bh^a. \nonumber \eeq

\noindent which implies that the infinitesimal generator, ${{\bf X}}$, of $\Phi$ must satisfy
\beq \mathcal{L}_{{\bf X}} \bh^a = 0. \eeq

If the solution admits a non-trivial isotropy group, we are able to {\it prolong} the manifold \cite{olver1995} to produce a larger manifold and determine an invariant coframe. This is achieved by appending the parameters of the remaining Lorentz frame transformations that constitute the isotropy group of dimension $s = dim(H_p)$, and using the Maurer-Cartan one-forms for each of the frame transformation infinitesimal generators we find the frame basis for the new manifold satisfies \cite{olver1995}: 
\beq \mathcal{L}_X \bh^A = 0,~~A \in [1, n+s].  \eeq

Instead of prolongation we can continue to work with the original manifold. The effect of a pullback of a diffeomorphism is equivalent to the action of the Lorentz frame transformation subgroup $H_p$ acting on the frame:
\beq \Phi^{*} \bh^a = \Lambda^a_{~b} \bh^b,~~ \Lambda \in H_p. \nonumber \eeq 
\noindent the metric is insensitive to this fact since,
\beq \Phi^* \eta_{ab} \bh^a \bh^b = \eta_{ab} \bh^a \bh^b \eeq 
\noindent and so the Lie derivative of the metric yields the Killing equations, $\mathcal{L}_{{\bf X}} g_{ab} = 0$.  The metric does not detect the remaining linear isotropy; however, we can use the Killing equations to determine the set of potential affine frame symmetries. 

To determine those Killing vector fields which preserve the structure of the geometry we must return to the frame bundle, $\mathcal{F}(M)$ and the uniquely defined basis of the cotangent space, $\{{\bf \tilde{h}}^a, \tilde{\bomega}^a_{~b}\}$. For this basis, the induced action of $\Phi$ on $\mathcal{F}(M)$ satisfies \eqref{Fmap} and the infinitesimal generator of $\Phi$, ${\bf X}$ satisfies
\beq \mathcal{L}_{{\bf X}} \tilde{\bomega}^a_{~b} = {\bf 0}. \eeq

\noindent While this condition on the connection one-form is originally defined on the frame bundle, we can project it down to the original manifold, so that \cite{fon1992,kramer}:
\beq \mathcal{L}_{{\bf X}} \bomega^a_{~b} = {\bf 0}. \eeq
\noindent where $\bomega^a_{~b}$ is defined in terms of the frame determined from the algorithm. 
 
From these observations we can now provide definite criteria to determine affine frame symmetries for a teleparallel geometry:

\begin{prop} For any solution of teleparallel gravity, an affine frame symmetry is a Killing vector field, ${\bf X}$, 
(i.e.,  $\mathcal{L}_{{\bf X}} {\bf g} = {\bf 0}$) that also annihilates the spin-connection relative to the frame determined by the Cartan-Karlhede algorithm: 
\beq \mathcal{L}_{{\bf X}} \omega^a_{~bc}= 0. \nonumber \eeq
\end{prop}

We conclude this section with a simple proof to show that there are no teleparallel geometries which admit $SO(1,3)$ as an isotropy group except Minkowski space, this confirms and extends the result of \cite{HJKP2018}.

\begin{thm} \label{thm:MaxSymFS}
There are no teleparallel geometries admitting a maximal group of affine frame symmetries other than Minkowski space.
\end{thm}

\begin{proof}
In order to obtain a ten dimensional isometry group for a teleparallel geometry, the discrete sequences must be 
\beq \{t_p\}= \{ 0, 0 \},~~\{ dim~H_p\} = \{ 6, 6 \} \nonumber \eeq
\noindent Then the group of affine frame symmetries is 
\beq r = s+D - t_p = 6+4-0 = 10 \nonumber \eeq
\noindent To attain the sequence for $\{ t_p\}$, the components of the torsion tensor must be constant, which is permissible. However, the second condition can only occur when the axial, vector and purely tensorial parts of torsion vanish. Hence, the space must have vanishing torsion and curvature.
\end{proof}

\noindent In subsection \ref{subsec:tensortorsion} we will show that the existence of a non-zero torsion tensor restricts the maximum dimension of the group of affine frame symmetries in {\it any} Riemann-Cartan geometry.

\newpage
\section{Classification of the torsion tensor and its covariant derivatives} \label{sec:TorsionDecomp}

Unlike the classification of the Ricci and Weyl tensor in metric-based theories, we cannot treat the torsion tensor as an operator to determine its eigenvalues and eigenbasis. In principle, the alignment classification can be applied to the torsion tensor \cite{classa, classb, classc}. However, the alignment classification is not fine enough to distinguish between important inequivalent torsion tensors; for example, it is insensitive to the linear isotropy groups of the torsion tensor. 

In order to classify the torsion tensor and determine the linear isotropy group of this tensor, we propose the following two part approach. We will first classify the linear isotropy group of the irreducible parts of the torsion tensor. In each case we will provide a canonical form of the irreducible parts of the torsion tensor and indicate the group of Lorentz transformations that can still be used \cite{hall2004symmetries}. Once this has been done, we will then apply the alignment classification to distinguish between subcases. As such the classification of the torsion tensor will be labelled by the isotropy group and then by its alignment type.

\subsection{Irreducible parts of the torsion tensor} 


The torsion two-form, $T^a$, can be expanded as
\begin{equation}
T^a=\frac{1}{2}T^a_{\phantom{a}bc} \,h^b \wedge h^c
\end{equation}
with $24$ independent components. Under the global Lorentz group, the torsion tensor can be decomposed into three irreducible parts \cite{Hehl_McCrea_Mielke_Neeman1995}: 
\beq T_{abc} = \frac23 (t_{abc} - t_{acb}) - \frac13 (g_{ab} V_c - g_{ac} V_b) + \epsilon_{abcd} A^d.\label{TorsionDecomp} \eeq

\noindent Here ${\bf V}$ denotes the vector part which is the trace of the torsion tensor: 
\beq V_a = T^b_{~ba}, \label{Vtor} \eeq
\noindent Lowering the index of the torsion tensor and applying the Hodge dual of the resulting tensor gives the axial part, ${\bf A}$: 
\beq A^a = \frac16 \epsilon^{abcd}T_{bcd}. \label{Ator} \eeq
\noindent Finally, we can construct the purely tensorial part, ${\bf t}$: 
\beq t_{(ab)c} = \frac12 (T_{abc}+ T_{bac}) -\frac16 (g_{ca} V_b + g_{cb} V_a) + \frac13 g_{ab} V_c. \label{Ttor} \eeq
\noindent We will call each of these tensors the {\it vector part , axial part , and tensor part } of the torsion tensor. The tensor part satisfies the following identities: 
\beq \begin{aligned}
& g^{ab} t_{(ab)c} = 0, ~t_{(ab)c} = t_{(ba)c},~t_{(ab)c} + t_{(bc)a} + t_{(ca)b} = 0. \end{aligned} \eeq

\noindent A simple counting argument shows that the tensor part of the torsion will have 16 algebraically independent components. Including the components of the vector part and axial part, this gives 24 components for the torsion tensor. 

Working with a complex null frame $ \{ \bh^a \} = \{ \bn, \bell, \bm, \bar{\bm} \}$ from the orthonormal frame $\{ \tilde{\bh}^a \} = \{ {\bf u}, {\bf x}, {\bf y}, {\bf z} \}$ through the transformation
\beq \begin{aligned} & \bn = \frac{1}{\sqrt{2}} ( {\bf u} + {\bf x}),~~ \bell = \frac{1}{\sqrt{2}} ( {\bf u} - {\bf x}),~\bar{\bm} = \frac{1}{\sqrt{2}} ( {\bf y} + i {\bf z}),~ \bm  = \frac{1}{\sqrt{2}} ( {\bf y} - i {\bf z}), \end{aligned} \label{orthoframe} \eeq
\noindent we can choose a basis of algebraically independent components: 
\beq \begin{aligned} & t_{(12)1},~ t_{(13)1},~ t_{(14)1},~ t_{(22)1},~ t_{(23)1},~ t_{(24)1},~ t_{(33)1},~ t_{(44)1},~ \\
& t_{(13)2},~ t_{(14)2},~ t_{(23)2},~ t_{(24)2},~ t_{(33)2},~ t_{(44)2},~ t_{(14)3},~ t_{(24)3},  \end{aligned}\eeq

\noindent while the remaining components have the following algebraic dependencies from the trace-free and cyclic identities, respectively:
\beq \begin{aligned} & t_{(34)1} = t_{(12)1},~ t_{(34)2} = t_{(12)2},~ t_{(34)3} = t_{(12)3},~  t_{(34)4} = t_{(12)4}, \\
\end{aligned} \eeq
\noindent and
\beq \begin{aligned} 
& t_{(11)2} = - t_{(12)1},~ t_{(11)3} = - t_{(13)1},~ t_{(11)4} = - t_{(14)1},~ t_{(12)2} = - t_{(22)1},\\
&  t_{(13)3} = - t_{(33)1},~ t_{(14)4} = -t_{(44)1},~ t_{(22)3} = -t_{(23)2},~ t_{(23)3} = -t_{(33)2},\\
&  t_{(22)4} = -t_{(24)2},~ t_{(24)4} = -t_{(44)2},~ t_{(33)4} = -t_{(12)3},~t_{(12)4} = -t_{(44)3},\\
&  t_{(12)3} = - t_{(23)1} -t_{(13)2},~ t_{(44)3} = t_{(24)1} +t_{(14)2},  \\
&  t_{(13)4} = -t_{(12)1} - t_{(14)3},~ t_{(23)4} = t_{(22)1} - t_{(24)3},~ \\ 
& t_{(11)1} = 0,~ t_{(22)2} = 0,~ t_{(33)3} = 0,~ t_{(44)4} = 0. \end{aligned}  \eeq

We emphasize that the above relationship between the components of the purely tensorial part of the torsion tensor has been presented in an arbitrary complex-null frame.

\subsection{Permitted linear isotropy groups of the vector and axial parts of torsion} \label{subsec:vectortorsion}

The axial part, ${\bf A}$, and vector part, ${\bf V}$, of the torsion tensor can readily be classified using their magnitude and direction. When these vector fields are non-vanishing, it is natural to adapt the frame to them. Even in the case that ${\bf A}={\bf V}$, we can adapt the frame depending on whether ${\bf V}$ is timelike or spacelike by choosing, respectively, \beq \tilde{\bh}^1 = {\bf V}/\sqrt{-|{\bf V}|^2}  \text{ if } |{\bf V}|^2 < 0 \text{ or } \tilde{\bh}^2 = {\bf V}/\sqrt{|{\bf V}|^2} \text{ if } |{\bf V}|^2 > 0,\eeq  
\noindent while if $| {\bf V}|^2 = 0$ we can choose ${\bf V}$ to be an element of a null frame basis. 

\begin{enumerate}

\item ${\bf A} = a {\bf V}$ , $a \in \mathbb{R}$

\begin{enumerate}

\item $|{\bf V}| < 0$: The frame can be adapted so that ${\bf V} = v \tilde{\bh}^1$, $v \in \mathbb{R}$. The linear isotropy group is $SO(3)$. 
\item  $|{\bf V}| = 0$: The frame can be adapted so that ${\bf V} = v(\tilde{\bh}^1 - \tilde{\bh}^2) = \sqrt{2} v \bell$, $v \in \mathbb{R}$. The linear isotropy group is $E(2)$, and boosts can be used to set $v=1$. 
\item $|{\bf V}| > 0$: The frame can be adapted so that ${\bf V} = v \tilde{\bh}^2$, $v \in \mathbb{R}$. The linear isotropy group is $SO(1,2)$. 

\end{enumerate}

\item ${\bf A} \neq a {\bf V}$ 

Denoting $\bA' \in span(\bV)^\perp$, 

\begin{enumerate}

\item $|{\bf V}| <0  \text{ and } |{\bf A}'| >0$: Adapt the frame so that
\beq \bV = v \tilde{\bh}^1,~\bA = a_1 \tilde{\bh}^1 + a_2 \tilde{\bh}^2. \eeq
\noindent The remaining frame freedom consists of rotations about 
$\tilde{\bh}^2$.

\item $|\bV| = 0 \text{ and } |\bA'| <0 $
\beq \bV = v (\tilde{\bh}^1-\tilde{\bh}^2),~\bA = a_1 (\tilde{\bh}^1-\tilde{\bh}^2)+ a_2 \tilde{\bh}^1. \eeq
\noindent The remaining frame freedom consists of 1D null rotations about $\bV$. 

\item $|\bV| = 0 \text{ and } |\bA'| =0 $
\beq \bV = v (\tilde{\bh}^1-\tilde{\bh}^2),~\bA = a_1 (\tilde{\bh}^1 - \tilde{\bh}^2) + a_2 (\tilde{\bh}^1+\tilde{\bh}^2). \eeq
\noindent The remaining freedom consists of rotations in the $\tilde{\bh}^3$ - $\tilde{\bh}^4$ plane.

\item $|\bV| = 0 \text{ and } |\bA'| >0 $
\beq \bV = v (\tilde{\bh}^1-\tilde{\bh}^2),~\bA = a_1 (\tilde{\bh}^1-\tilde{\bh}^2) + a_2 \tilde{\bh}^3 . \eeq
\noindent The remaining freedom consists of 1D null rotations about $\bV$.

\item  $|{\bf V}| >0  \text{ and } |{\bf A}'| <0$: Adapt the frame so that 
\beq \bV = v \tilde{\bh}^2,~\bA = a_1 \tilde{\bh}^2 + a_2 \tilde{\bh}^1. \eeq
The remaining frame freedom consists of rotations about 
$\tilde{\bh}^2$. 

\item  $|{\bf V}| >0  \text{ and } |{\bf A}'| =0$: Adapt the frame so that 
\beq \bV = v \tilde{\bh}^2,~\bA = a_1 \tilde{\bh}^2 + a_2( \tilde{\bh}^1 - \tilde{\bh}^3). \eeq
\noindent The remaining frame freedom consists of one-dimensional (1D) null rotations about $\bA'$.

\item  $|{\bf V}| >0  \text{ and } |{\bf A}'| >0$: Adapt the frame so that 
\beq \bV = v \tilde{\bh}^2,~\bA = a_1 \tilde{\bh}^2 + a_2 \tilde{\bh}^3. \eeq
\noindent The remaining frame freedom consists of boosts in the $\tilde{\bh}^1$ - $\tilde{\bh}^4$ plane. 

\end{enumerate}

\end{enumerate}

\subsection{Permitted linear isotropy groups of the purely tensorial part of torsion} \label{subsec:tensortorsion}

The remaining linear isotropy at zeroth order is determined by the tensor part of the torsion. Using the infinitesimal generators of the Lorentz frame transformation group and knowledge of all permitted subgroups \cite{hall2004symmetries, Milson:2007ft} it is possible to determine canonical forms for the tensor-part of the torsion. Due to the decomposition of the torsion tensor \eqref{TorsionDecomp}, we will consider the anti-symmetrized tensor: 
\beq \hat{T}_{abc} = \frac23 (t_{(ab)c}-t_{(ac)b}) \label{That} \eeq

\noindent and relate the linear isotropy group of $\hat{T}_{abc}$ to the non-zero components of $t_{(ab)c}$. To describe the canonical form of $t_{(ab)c}$, we list the non-zero components for each of the permitted one dimensional (1D), 2D  or 3D linear isotropy groups. The canonical forms listed here allow for a particular choice of the frame determined by the linear isotropy group of the torsion tensor. 

\begin{itemize}
\item {\bf 1D linear boost isotropy group}:

The canonical form of the torsion tensor invariant under transformations of the form,
\beq & \bell' = D^2 \bell,~~{\bf n'} = D^{-2} {\bf n}, {\bf m}'=  {\bf m},~~ D, \in \mathbb{R}, \label{1DBoost} & \eeq
\noindent has the following non-zero components:
\beq t_{(23)1}, t_{(24)1}, t_{(13)2}, t_{(14)2}.  \label{BoostIso} \eeq

\noindent Null rotations about $\bell$ and $\bn$ are excluded as they will change the form of the canonical form. Spatial rotations  can be used to  while spins are permitted to set one of the non-zero components to equal another. For example, a spin can be applied to set $t_{(14)2} = t_{(13)2}$. 

\item {\bf 1D linear spin isotropy group}: 

The canonical form of the torsion tensor invariant under transformations of the form,
\beq & \bell' = \bell,~~{\bf n'} =   {\bf n}, {\bf m}'= e^{i \theta} {\bf m},~~ \theta \in \mathbb{R}, \label{1Dspin} &\eeq

\noindent has the following non-zero components
\beq t_{(12)1}, t_{(14)3}, t_{(22)1},  t_{(24)3}, t_{(22)1}. \label{SpinIso} \eeq


\noindent In this case, null rotations about $\bell$ and $\bn$ are excluded, while boosts are permitted to set one of the non-zero components to some non-zero constant. For example, a boost can lead to $t_{(12)1} = 1$.

\item {\bf 1D linear null rotation isotropy group}: 

The canonical form of the torsion tensor invariant under null rotations about $\bell$,
\beq & \bell' = \bell,~~{\bf n}' = {\bf n} + B + {\bf m} + B \bar{{\bf m}} + B^2  \bell, {\bf m}' = {\bf m} + B \bell,~~ B \in \mathbb{R}, & \eeq

\noindent has non-zero components:
\beq t_{(22)1},~t_{(33)2} = t_{(22)1}+t_{(24)3},~ t_{(44)2} = 2 t_{(22)1} - t_{(24)3},~ t_{(24)3} ,~ t_{(23)2}. \label{1DNRIso}\eeq

\noindent As null rotations about $\bn$ and null rotations about $\bell$ with a purely complex parameter, $iB, B \in \mathbb{R}$, affect the form of the torsion tensor, they are excluded. Boost and spins can be used to set some components to constant or equal other components. For example,  $t_{(22)1} = 1$ and $t_{(14)2} = t_{(13)2}$.

\item {\bf 2D linear null rotation isotropy group}:

The canonical form of the torsion tensor invariant under 
\beq & \bell' = \bell,~~{\bf n}' = {\bf n} + B + {\bf m} + \bar{B} \bar{{\bf m}} + B \bar{B} \bell, {\bf m}' = {\bf m} + B \bell,~~ B \in \mathbb{C}, & \eeq

\noindent has non-zero components:
\beq t_{(23)2}, t_{(24)2}.  \label{2DNRIso} \eeq

\end{itemize}

\noindent Null rotations about $\bn$ affect the form of the torsion tensor and hence they are excluded. Boost and spins can be used to set some components to constant or equal another component. For example,  $t_{(23)2} = 1$ and $t_{(24)2} = 0$.

\subsection{Frame fixing and the alignment classification}

Unlike the vector fields, the tensor part of the torsion is not easily classified; however, it is possible to apply the alignment classification, using the effect on the tensor by a boost on a chosen null coframe $\{ \bn, \bell, \bm^i\}$, $\bell' = \lambda \bell,~~ {\bf n'} = \lambda^{-1} \bn$. For an arbitrary tensor, ${\bf \bar{T}}$, of rank $n$ the components transform under a boost as
\beq \bar{T}'_{a_1 a_2...a_n} = \lambda^{b_{a_1 a_2 ... a_n}} \bar{T}_{a_1 a_2 ... a_n},~~ b_{a_1 a_2...a_n} = \sum_{i=1}^n(\delta_{a_i 1} - \delta_{a_i 2}),  \eeq
\noindent where $\delta_{ab}$ denotes the Kronecker delta symbol. The quantity, $b_{a_1 a_2 ... a_n}$, is called the {\it boost weight} (b.w.) of the frame component $\bar{T}_{a_1 a_2 ... a_p}$. The b.w. of each algebraically independent component of $t_{(ab)c}$ and $t_{(ab)c|d}$, defined relative to an arbitrary frame, are displayed in figure 1 and figure 2, respectively. 

\begin{figure} \label{fig:bwd1torsion}
\beq \begin{array}{c|c} b.w. & Components \\ \hline
2 & t_{(13)1}, t_{(14)1} \\
1 & t_{(12)1}, t_{(33)1}, t_{(44)1} t_{(14)3} \\
0 & t_{(23)1}, t_{(24)1}, t_{(13)2}, t_{(14)2} \\
-1 & t_{(22)1}, t_{(33)2}, t_{(44)2}, t_{(24)3} \\
-2 & t_{(23)2}, t_{(24)2}  \end{array} \eeq
\caption{The b.w. values of the algebraically independent components of the purely tensorial part of torsion in an arbitrary frame.}
\end{figure} 

\begin{figure} \label{fig:bwtorsion}
\beq \begin{array}{c|c} b.w. & Components \\ \hline
3 & t_{(13)1|1}, t_{(14)1|1} \\
2 & t_{(13)1|3}, t_{(13)1|4}, t_{(14)1|3}, t_{(14)1|4},  t_{(12)1|1}, t_{(33)1|1}, t_{(44)1|1}, t_{(14)3|1} \\
1 & t_{(13)1|2}, t_{(14)1|2}, t_{(12)1|3}, t_{(12)1|4},    t_{(33)1|3}, t_{(33)1|4},  t_{(44)1|3}, t_{(44)1|4}, \\
 & t_{(14)3|3},  t_{(14)3|4},  t_{(23)1|1}, t_{(24)1|1},  t_{(13)2|1}, t_{(14)2|1} \\
0 & t_{(12)1|2}, t_{(33)1|2}, t_{(44)1|2}, t_{(14)3|2},  t_{(23)1|3}, t_{(23)1|4}, t_{(24)1|3}, t_{(24)1|4}, \\ 
&  t_{(13)2|3}, t_{(13)2|4}, t_{(14)2|3}, t_{(14)2|4},  t_{(22)1|1}, t_{(33)2|1}, t_{(44)2|1}, t_{(24)3|1} \\
-1 & t_{(23)1|2}, t_{(24)1|2}, t_{(13)2|2}, t_{(14)2|2},  t_{(22)1|3}, t_{(22)1|4}, t_{(33)2|3}, t_{(33)2|4},\\ 
&  t_{(44)2|3}, t_{(44)2|4}, t_{(24)3|3}, t_{(24)3|4}  t_{(23)2|1}, t_{(24)2|1} \\ 
-2 & t_{(22)1|2},  t_{(33)2|2}, t_{(44)2|2}, t_{(24)3|2},  t_{(23)2|3}, t_{(23)2|4}, t_{(24)2|3}, t_{(24)2|4}, \\
-3 & t_{(23)2|2}, t_{(24)2|2}   
\end{array} \eeq
\caption{The b.w. values of the algebraically independent components of $t_{(ab)c|d}$ in an arbitrary frame.}
\end{figure}

We define the {\it boost order}, $\mathcal{B}_{{\bf \bar{T}}}(\bell)$, as  the maximum b.w. of a tensor, ${\bf \bar{T}}$, relative to a particular null direction $\bell$. We will consider an arbitrary tensor, ${\bf \bar{T}}$, with boost order at most 2 for all choices of the null direction $\bell$, such as the torsion tensor and any rank two tensor. The tensor, ${\bf \bar{T}}$, can be broadly classified into five {\it alignment types}: {\bf G},{\bf I}, {\bf II}, {\bf III}, and {\bf N} if there exists an $\bell$ such that $\mathcal{B}_{{\bf \bar{T}}} (\bell) = 1, 0,-1,-2$ and we will say $\bell$ is ${\bf \bar{T}}$-aligned, while if {\bf T} vanishes, then it belongs to alignment type {\bf O}. We will define the {\it secondary alignment type} as the alignment classification of the remaining null direction $\bn$. 

For higher rank tensors, like the covariant derivatives of the torsion tensor, the alignment types are still applicable although it is possible that $|\mathcal{B}_{{\bf \bar{T}}} (\bell)|$ may be greater than two. In Riemannian geometry, any spacetime that cannot be characterized by the set $\mathcal{I}$, known as an $\mathcal{I}$-degenerate spacetime, admits a null coframe such that all of the positive b.w. terms of the curvature tensor and its covariant derivatives are zero in this common coframe; that is, they are all of alignment type {\bf II} \cite{4dCSI}.

\subsubsection{Algorithmic approach to the choice of the invariant frame}

To determine the invariant frame, we will first exploit the existence of the vector and axial parts of the torsion, $\bV$ and $\bA$, to potentially fix two of the frame basis elements. In some of the possible cases for $\bV$ and $\bA$ there is some freedom in the choice of this frame. In order to fix the frame in line with the alignment classification we will always choose the frame that minimizes the alignment type of the resulting form of the torsion tensor. 

Once the frame has been adapted to $\bV$ and $\bA$ as much as possible, the remaining frame freedom is used to put the purely tensorial part of the torsion tensor, $t_{(ab)c}$, into a canonical form. If $t_{(ab)c}$ admits the same linear isotropy as  $\bV$ and $\bA$, then the associated canonical form for $t_{(ab)c}$ can be realized in the same frame. Otherwise, we will use the canonical form arising from the alignment classification.

We remind the reader that there are other frames that can be used. One notable example is the proper frame in which the spin-connection is trivial, $\omega^a_{~bc} = 0$. This frame is not an invariant frame as defined above since it cannot generally be defined using tensorial conditions.  

\subsubsection{Alignment classification of the trace and vector parts of torsion}

By adapting the frame to the vector or axial parts of the torsion tensor, the alignment classification of their contribution to the torsion tensor will be determined: 
\beq \hat{T}_{abc} = \frac23 g_{a[b}V_{c]} + \epsilon_{abcd}A^d \eeq

\noindent We emphasize that this choice of frame will have no impact on the purely tensorial part, $t_{(ab)c}$, since it is an irreducible piece of the torsion tensor. 

We note that in some cases when one of the invariant vectors is spacelike, we have the choice of adapting the frame so that the resulting tensor, $\hat{T}_{abc}$, will be either of type {\bf I} and secondary type {\bf I} or type {\bf II}. While it appears to be advantageous to fix the frame to ensure the type {\bf II} property, we will include the other option for the sake of completeness. 

\begin{enumerate}

\item ${\bf A} = a {\bf V}$ , $a \in \mathbb{R}$

\begin{enumerate}

\item $|{\bf V}| < 0$ 

Type {\bf I} and secondary type {\bf I} since ${\bf V} \propto \bell - \bn$. 

\item $|{\bf V}| > 0$: 

\begin{itemize}
\item Type {\bf I} and secondary type {\bf I} by choosing ${\bf V} \in span(\bell, \bn)$
\item Type {\bf II} if ${\bf V} \notin span(\bell, \bn)$
\end{itemize}

\item  $|{\bf V}| = 0$: 

Type {\bf III} since $\bV \propto \bell$. 
\end{enumerate}

\item ${\bf A} \neq a {\bf V}$ 

In this case we are able to define an invariant two-plane by adapting ${\bf V}$ and $\bA' \in span(\bV)^\perp$, the orthogonal part of $\bA$. We will classify the contributions of the axial and vector parts of the torsion tensor separately:

\beq \hat{T}^{V}_{abc} = \frac23 g_{a[b}V_{c]}, \hat{T}^A_{abc} =  \epsilon_{abcd}A^d \label{AVtensors} \eeq

\begin{enumerate}

\item $|{\bf V}| <0  \text{ and } |{\bf A}'| >0$

$\hat{T}^{V}_{abc}$ is of type {\bf I} and secondary type {\bf I}.

\begin{itemize}
\item If $\bV \cdot \bA = 0$ either $\hat{T}^{A}_{abc}$ will be of type {\bf I} and secondary  type {\bf I} or type {\bf II}.
\item If $\bV \cdot \bA \neq 0$ either $\hat{T}^{A}_{abc}$ will be of type {\bf I} and secondary  type {\bf I}
\end{itemize}

\item $|\bV| = 0 \text{ and } |\bA'| <0 $

$\hat{T}^{V}_{abc}$ is of type {\bf III}. 

$\hat{T}^{A}_{abc}$ will be of type {\bf I} and secondary type {\bf I}.

\item $|\bV| = 0 \text{ and } |\bA'| =0 $

Choosing $\bV \propto \bell$ and $\bA \propto \bn$ this implies that the sum $\hat{T}_{abc} = \hat{T}^{V}_{abc} + \hat{T}^{A}_{abc}$ will be of type {\bf I} and secondary type {\bf I}.

\item $|\bV| = 0 \text{ and } |\bA'| >0 $

$\hat{T}^{V}_{abc}$ is of type {\bf III}. 

$\hat{T}^{A}_{abc}$ will be of type {\bf II}.

\item  $|{\bf V}| >0  \text{ and } |{\bf A}'| <0$

$\hat{T}^{V}_{abc}$ is either of type {\bf II} or of type {\bf I} and secondary type {\bf I}. 

$\hat{T}^{A}_{abc}$ will be of type {\bf I} and secondary type {\bf I}.

\item  $|{\bf V}| >0  \text{ and } |{\bf A}'| =0$

\begin{itemize}
\item If $\bV \cdot \bA = 0$ then $\hat{T}^{A}_{abc}$ is of type {\bf III}. In this case $\hat{T}^{V}_{abc}$ must be of type {\bf II}.  
\item If $\bV \cdot \bA \neq 0$ then both $\hat{T}^{V}_{abc}$ and  $\hat{T}^{A}_{abc}$ are of  type {\bf I} and secondary type {\bf I}.
\end{itemize}

\item  $|{\bf V}| >0  \text{ and } |{\bf A}'| >0$

$\hat{T}^{V}_{abc}$ is either of type {\bf II} or of type {\bf I} and secondary type {\bf I}. 

\begin{itemize}
\item If $\bV \cdot \bA = 0$ then $\hat{T}^{A}_{abc}$ is of type {\bf II}. 
\item If $\bV \cdot \bA \neq 0$ then $\hat{T}^{A}_{abc}$ is either of type {\bf II} or of type {\bf I} and secondary type {\bf I} depending on the choice of type for $\hat{T}^{V}_{abc}$.
\end{itemize}

\end{enumerate}

\end{enumerate}

\subsubsection{Alignment classification of the canonical forms for the purely tensorial part of torsion}

\begin{itemize}

\item {\bf Trivial linear isotropy}

In this case, the torsion tensor fully specifies an invariant frame. This eliminates alignment types {\bf D}, {\bf N} as tensors of these two types always admit non-trivial linear isotropy groups. The permitted alignment types are {\bf I}, {\bf II} and {\bf III}. 

\item {\bf 1D boost isotropy:}

The canonical form of the torsion tensor must be of type {\bf D}. 

\item {\bf 1D spin isotropy: }

The canonical form of the torsion tensor is of alignment type {\bf I} and secondary alignment type {\bf I}. 

\item {\bf 1D null rotation isotropy: }

The canonical form of the torsion tensor is of alignment type {\bf III}. 

\item {\bf 2D null rotation isotropy: }

The canonical form of the torsion tensor is of alignment type {\bf N}. 

\item {\bf 3D $SO(3)$ isotropy:}

In this case, the purely tensorial part of the torsion must vanish and $\bA = a \bV$ with $\bV = v \bh^1$. The form of the torsion tensor is then

\beq T_{abc} = -\frac13 (g_{ab} V_c - g_{ac} V_b) + \epsilon_{abcd} A^d \eeq

\noindent As $\sqrt{2} \bh^1 = v( \bell + \bn)$ this implies the torsion tensor must be of alignment type {\bf I} and secondary alignment type {\bf I}. 

\item {\bf 3D $SO(1,2)$ isotropy:}

The purely tensorial part of the torsion must vanish and $\bA = a \bV$ with $\bV = v \bh^2$. Since $\sqrt{2} \bh^1 = v( \bell - \bn)$ this implies the torsion tensor must be of alignment type {\bf I} and secondary alignment type {\bf I}. 

\item {\bf 3D $E(2)$ isotropy:}

In this case, the purely tensorial part of the torsion must vanish and $\bA = a \bV$ with $\bV = v \bell$. The canonical form of the torsion tensor is of alignment type {\bf III} and secondary alignment type {\bf II}.

\end{itemize}

\subsection{Classification of the covariant derivative of the torsion tensor}

We must classify the first covariant derivative of the torsion tensor in order to fully classify the field equations of a teleparallel gravity theory. To do so we will apply the alignment classification for each of the possible torsion tensors with non-trivial linear isotropy group.

\begin{itemize}

\item {\bf 1D boost isotropy:}

The possible alignment types are {\bf II} or {\bf D}.

\item {\bf 1D spin isotropy: }

In the case, the alignment type can be {\bf G} or {\bf I} and secondary alignment type {\bf G} or {\bf I}. 

\item {\bf 1D null rotation isotropy: }

The alignment types are {\bf II} or {\bf III}. 

\item {\bf 2D null rotation isotropy: }

The alignment types are {\bf III} or {\bf N}.

\item {\bf 3D $SO(3)$ isotropy:}

The alignment types are {\bf O} or alignment type  {\bf G} or {\bf I} and secondary alignment type {\bf G} or {\bf I}.

\item {\bf 3D $SO(1,2)$ isotropy:}

The alignment types are {\bf O} or alignment type  {\bf G} or {\bf I} and secondary alignment type {\bf G} or {\bf I}.

\item {\bf 3D $E(2)$ isotropy:}

The alignment types in this case are {\bf O}, {\bf II} or {\bf III} with secondary alignment type {\bf I} or {\bf II}.

\end{itemize}

\subsection{Degenerate Kundt frame}

A class of teleparallel geometries of particular interest are those in which there exists a frame, called a degenerate Kundt frame, relative to which all of the positive b.w. components of the torsion tensor and its covariant derivatives are simultaneously zero (i.e., they are all of algebraic type {\bf II} in the same frame). If such a degenerate Kundt frame exists, then ${\bf V}$ and ${\bf A}$ can be invariantly fixed according to cases 1.b, 1.c, 2.a, 2.b, 2.d or 2.e in subsection \ref{subsec:vectortorsion}. This may not be an invariant frame as the torsion tensor can still admit a non-trivial linear isotropy group. 

This choice of frame leads to constraints on the components of the torsion tensor as the contributions to the torsion tensor from the axial and vector torsion tensors in \eqref{AVtensors} will be of alignment type {\bf II} at most. If a frame is chosen in this way, the permitted group of Lorentz frame transformations will be reduced as well, and this can affect the choice of canonical form for the purely tensorial part of the torsion tensor.

In this frame, all of the positive b.w. terms for the irreducible torsion tensor and its first covariant derivative as given in figures 1 and 2 must be zero. This will lead to a number of constraints on the frame and the spin-connection as the frame may not be proper. We will return to this in future work.

\subsection{Maximal dimension of the affine frame symmetry group}

As a corollary of the classification of the linear isotropy groups of the purely torsion part of the tensor, we have a related result that is applicable for {\it any} Riemann-Cartan geometry with a non-zero torsion tensor \cite{fon1992} 

\begin{cor} \label{cor:maxsym}
If a Riemann-Cartan geometry admits a non-zero torsion tensor, then the maximum dimension of the group of affine symmetries is at most seven. 
\end{cor}

\begin{proof}
Regardless of the structure of the curvature tensor, if the torsion tensor is non-zero, the linear isotropy group of the torsion tensor is at most three. Following the more general equivalence algorithm in \cite{fon1992} this implies that the largest isometry group, $$ r = 4-t_p + dim~H_p$$ 
will occur when $t_p = 0$ and $dim~H_p = 3$. 
\end{proof}

\subsection{Highest iteration of the modified Cartan-Karlhede algorithm}

As the dimension of the continuous isotropy group permitted for a teleparallel geometry is at most 3D, the maximum dimension of the reduced frame bundle in this procedure is q=7.
There is one case where a teleparallel geometry has the potential to attain the highest iteration in the algorithm. At zeroth order the torsion tensor must satisfy the following conditions:

\begin{enumerate}
\item The purely tensorial part of the torsion tensor vanishes.
\item The axial and vector parts of the torsion tensor are proportional to each other. 
\item The proportionality factor is constant, $\bV = c \bA$, $c \in \mathbb{R}$. 
\end{enumerate}

At each iteration, $q>0$, we must then have either $t_q$ or $dim~H_q$ iterate by at most one while the other discrete quantity is unchanged. At the end of the algorithm, $t_7 = 4$ and $dim~H_7=0$. The corresponding group of affine frame symmetries is trivial for any teleparallel geometry for which $q=7$.


\newpage
\section{Examples of Teleparallel Geometries} \label{sec: Examples}

To illustrate the modified Cartan-Karlhede algorithm, we will consider several teleparallel geometries. The examples presented arise from the choice of a particular frame associated with a given spacetime metric and the choice of a spin-connection.  In these examples, the spin-connection $\omega^a_{~bc}$ is taken to be anti-symmetric in the first two indices and satisfies the constraints given by \eqref{zero-curvature} and \eqref{temp2}. We note that in examples A-D there is a non-trivial linear isotropy group at the termination of the algorithm, and this is inherited by the spin-connection. 

\subsection{Friedmann-Lemaitre-Robertson-Walker (FLRW) spacetime}

The coframe for the orthonormal diagonal frame gauge are \cite{Hohmann:2018rwf}:
\begin{equation}
\bh^a = \left[\begin{array}{c}
dt \\
R(t)\frac{1}{\xi}\,dr \\
R(t)r\,d\theta \\
R(t)r\sin(\theta)\,d\phi
\end{array} \right]
\end{equation}

\noindent where $\xi(r) = \sqrt{1-kr^2}$. We note that the connection one-forms are consistent with the anti-symmetric part of the $f(T)$ field equations are 
\begin{equation}
\bomega^a_{\ b} = \left[\begin{array}{cccc}
0 & 0 & 0 & 0 \\
0 & 0 & -\xi\,d\theta - \sqrt{k}r\sin(\theta)\,d\phi & \sqrt{k}r\,d\theta - \xi\sin(\theta)\,d\phi \\
0 & \xi\,d\theta +\sqrt{k}r\sin(\theta)\,d\phi & 0 & -\frac{\sqrt{k}}{\xi}\,dr-\cos(\theta)\,d\phi \\
0 & -\sqrt{k}r\,d\theta + \xi\sin(\theta)\,d\phi & \frac{\sqrt{k}}{\xi}\,dr+\cos(\theta)\,d\phi & 0
\end{array}\right].
\end{equation}

\noindent Computing the torsion tensor, we find that
\begin{equation}
\bT^a = \left[\begin{array}{c}
0 \\
\frac{\dot{R}(t)}{R(t)}\,h^1 \wedge h^2 + 2\frac{\sqrt{k}}{R(t)} \, h^3 \wedge h^4 \\
\frac{\dot{R}(t)}{R(t)}\,h^1 \wedge h^3 + 2\frac{\sqrt{k}}{R(t)} \, h^4 \wedge h^2 \\
\frac{\dot{R}(t)}{R(t)}\,h^1 \wedge h^4 + 2\frac{\sqrt{k}}{R(t)} \, h^2 \wedge h^3 \\
\end{array} \right].
\end{equation}

\noindent The vector and axial part of the torsion are, respectively,
\beq {\bf V}=  \frac{3\dot{R}}{R} h^1 \text{ and } {\bf A} = \frac{2 \sqrt{k}}{R} h^1. \eeq
\noindent while the purely tensorial part vanishes. As the frame has been adapted to $\bV$ alone, this is not an invariant frame yet; there may be additional tensors at higher order that will allow us to further fix the frame.

\noindent All three irreducible parts of the torsion are invariant under $SO(3)$, thus $dim~H_0 =3$ and there is only one functionally independent invariant, the magnitude of the axial part:
\beq \frac{2 \sqrt{k}}{R}. \eeq

\noindent We note that if $k = -1$ this leads to complex valued components of the torsion tensor. Thus for $k=-1$ the chosen connection is not ideal for a real-valued teleparallel geometry. It is possible to instead consider the connections in \cite{Hohmann:2015pva,Hohmann:2018rwf} to give real-valued FLRW geometries. In the case of $k \geq 0$ we can continue and note that the first covariant derivative retains the $SO(3)$ isotropy and does not include any new components. 

\subsection{The static spherically symmetric spacetime}

From \cite{Krssak_Saridakis2015} the static spherically symmetric line-element is 
\beq ds^2 = -A(r)^2 dt^2 + B(r)^2 dr^2 +r^2 d\theta^2 + r^2 \sin^2 \theta d\phi^2. \eeq
\noindent The corresponding coframe is  
\beq {\bf \tilde{h}}^a_{~\mu} = diag( A(r), B(r), r, r\sin\theta), \eeq

\noindent with the related half-null coframe:
\beq {\bh}^{a} = \left[ \begin{array}{c} \bn \\ \bell \\ {\bf X} \\ {\bf Y} \end{array} \right] = \left[ \begin{array}{c}  \frac{1}{\sqrt{2}} (A dt + B dr) \\ 
 \frac{1}{\sqrt{2}} (A dt - B dr) \\  rd\theta \\ r \sin\theta d\phi \end{array} \right].  \label{SSSnull} \eeq

\noindent Such that \beq ds^2 = -  \bell {\bf n} - {\bf n} \bell + {\bf X}^2 + {\bf Y}^2. \eeq

\noindent Choosing the inertial spin-connection, which is also consistent with the anti-symmetric part of the $f(T)$ field equations relative to the half-null frame:

\begin{equation}
\bomega^{b}_{~a} = \left[\begin{array}{cccc}
0 & 0 & \frac{1}{\sqrt{2}} d \theta & \frac{\sin \theta}{\sqrt{2}} d \phi  \\ 
0 & 0 & -\frac{1}{\sqrt{2}} d \theta & -\frac{\sin \theta}{\sqrt{2}} d \phi \\ 
- \frac{1}{\sqrt{2}} d \theta & \frac{1}{\sqrt{2}} d \theta & 0 & -\cos \theta d \phi  \\
-\frac{\sin \theta}{\sqrt{2}} d \phi & \frac{\sin \theta}{\sqrt{2}} d \phi & \cos \theta d \phi  & 0 \end{array} \right], \label{SSomega}
\end{equation}

\noindent then the torsion tensor has the following non-zero components
\beq T^a = \left[ \begin{array}{cccc} \frac{A'}{\sqrt{2}A B} {\bh}^1 \wedge {\bh}^2 \\
\frac{A'}{\sqrt{2}A B} {\bf h }^1 \wedge {\bh}^2 \\
\frac{B+1}{\sqrt{2} B r} ( {\bh}^1 \wedge {\bh}^3 - {\bh}^2 \wedge {\bh}^3) \\ 
\frac{B+1}{\sqrt{2} B r} ( {\bh}^1 \wedge {\bh}^4 - {\bh}^2 \wedge {\bh}^4)  \end{array} \right]. \eeq


\noindent The axial part vanishes, while the vector part of the torsion is: 
\beq {\bf V} = \frac{(A_{,r}r+2BA +2A)}{\sqrt{2}BAr} (\bell - \bn). \eeq

Looking at the form of the purely tensorial part, this tensor satisfies the conditions for spin isotropy, i.e. a $SO(1)$ rotation since
\beq t_{(11),2} = t_{(22),1} = -t_{(12),1}  = t_{(12),2} = - t_{(i3),3} = - t_{(i4),4} = t_{(33),i} = t_{(44),i},~~i \in[1,2]. \eeq

\noindent At zeroth order, there is one functionally independent component and the isotropy group consists of rotations in the $\bh^3$-$\bh^4 $ plane. The algorithm terminates at first order, as the isotropy group is unchanged and no new functionally independent invariants are produced. Thus $dim~H_1 = s =1$ and $t_1 = 1$ implying that $r = 1+4-1 = 4$.

\subsection{ Standard (Anti)-de Sitter spacetime} 

We will now consider the standard frame basis for the (anti)-de Sitter solution. In GR, the (anti)-de Sitter solution  can be described as special cases of the static spherically symmetric metric where $\Lambda$ is positive (negative) and the metric functions $A(r)$ and $B(r)$ are: $$A = B^{-1} = \sqrt{1-\frac{\Lambda r^2}{3}}.$$ Choosing the rest frame where the functions $A$ and $B$ are ``turned off'' (corresponding to the limit $\Lambda \to 0$), we find the expected spin-coefficients given in equation \eqref{SSomega}.  

The torsion tensor has the following non-zero components:
\beq \bT^a = \left[ \begin{array}{cccc} -\frac{\Lambda r}{\sqrt{2(9-3\Lambda r^2)}} {\bh}^1 \wedge {\bh}^2 \\
-\frac{\Lambda r}{\sqrt{2(9-3\Lambda r^2)}} {\bh}^1 \wedge {\bh}^2 \\
\frac{(3+ \sqrt{9-3\Lambda r^2})}{3r \sqrt{2}} ( {\bh}^1 \wedge {\bh}^3 - {\bh}^2 \wedge {\bh}^3) \\ 
\frac{(3+ \sqrt{9-3\Lambda r^2})}{3r \sqrt{2}}( {\bh}^1 \wedge {\bh}^4 - {\bh}^2 \wedge {\bh}^4)  \end{array} \right]. \eeq

\noindent As in the static spherically symmetric case, only the vector part of the torsion tensor is non-zero,
\beq {\bf V} =  \frac{ 6- 3\Lambda r^2 + 2 \sqrt{(9-3\Lambda r^2)}}{\sqrt{2r^2(9-3\Lambda r^2)}} (\bell - \bn). \eeq

\noindent The torsion scalar is non-constant, 
\beq T = \frac23 \left[ \frac{3\Lambda^2 r^4 - 2 (2 \Lambda r^2 -3 ) \sqrt{9-3\Lambda r^2} -15 \Lambda r^2 +18}{r^2 (\Lambda r^2 -3)}\right]. \eeq

At zeroth order, there is one functionally independent component and the isotropy group consists of rotations in the $\bh^3$-$\bh^4$ plane. The algorithm terminates at first order since $dim~H_0 = dim~H_1 = 1$ and $t_0 = t_1 = 1$. Thus the dimension of the isotropy group is still $r = 1+4-1 = 4$.

The fact that the number of functionally independent invariants is a significant difference from its analogue in GR. In the next example we will show the de Sitter metric admits other frames that will have a differing torsion structure and hence will attain a different symmetry group. We note that theorem \ref{thm:MaxSymFS} implies that no teleparallel geometry arising from the (anti-) de Sitter metrics will admit maximal symmetry.

\subsection{Two alternative tetrads for ``de Sitter'' spacetime}

Consider the line element for the de Sitter spacetime in coordinates $(t,x,y,z)$: 
\beq ds^2 = -dt^2 + R_0 e^{2H_0 t}(dx^2 + dy^2 + dz^2). \label{deSitterRob} \eeq

\noindent We will consider two potential frames: 
\beq \bh^a = \left[ \begin{array}{c} dt \\ R_0 e^{H_0 t} dx \\ R_0 e^{H_0 t} dy \\ R_0 e^{H_0 t} dz  \end{array} \right] \text{ and } \bh^{'a} = \left[ \begin{array}{c} \cosh(\alpha) dt + \sinh(\alpha)  R_0 e^{H_0 t} dx \\ \sinh(\alpha) dt + \cosh(\alpha)  R_0 e^{H_0 t} dx \\ R_0 e^{H_0 t} dy \\ R_0 e^{H_0 t} dz  \end{array} \right], \label{Robframe1} \eeq

\noindent where the boost angle is $\alpha = \alpha_0 R_0 H_0 e^{H_0 t } x$. We will assume the trivial spin-connection, $\bomega^a_{~b} = {\bf 0}$, for both frames. The second frame, $\bh^{'a}$, has been constructed by applying a Lorentz transformation to the frame, $\bh^{a}$, in \eqref{Robframe1} without applying the corresponding transformation to the spin-connection. Interestingly, both frames with the trivial connection give the same field equations. One frame is diagonal while the other is not and there is no Lorentz transformation mapping one frame into the other which also preserves the trivial connection condition. We will show that the two teleparallel geometries are not related by a diffeomorphism.

For the first tetrad, $\bh^{a}$, in \eqref{Robframe1} with $\bomega^a_{~b} = {\bf 0}$, the torsion tensor is of the form: 
\beq \bT^a = \left[ \begin{array}{c} 0 \\ H_0 \bh^1 \wedge \bh^2 \\ H_0 \bh^1 \wedge \bh^3 \\ H_0 \bh^1 \wedge \bh^4 \end{array} \right]. \eeq 

The axial and purely tensorial part of the torsion are zero, while the vector part is 
\beq {\bf V} = 3H_0 \bh^1. \eeq

\noindent The linear isotropy group must be $SO(3)$ in order to preserve the timelike direction $\bh^1$ and there are no functionally independent components, so $dim~H_0 = 3$ and $t_0 = 0$ At first order the covariant derivative of the torsion tensor is zero, thus $dim~H_1 = 3$ and $ t_1 = 0$. We conclude that the dimension of the symmetry group for this  teleparallel geometry is $r=3+4-0 = 7$.

To compute the torsion tensor for the second tetrad, $\bh^{'a}$, in \eqref{Robframe1} with $\bomega^a_{~b} = {\bf 0}$ we will translate the $t$ coordinate to set $\alpha_0 = 1$ and scale the coordinates so that the metric is conformal to a metric with $H_0 = 1$ with conformal factor $1/H_0$. Thus the boost angle is now $\alpha = R_0 x e^{t_0}$. The torsion tensor is 

\beq \bT^{'a} = \left[ \begin{array}{c} -\cosh(\alpha) \alpha \bh^{'1} \wedge \bh^{'2} \\ -\sinh(\alpha) \alpha \bh^{'1} \wedge \bh^{'2} \\ \cosh(\alpha) \bh^{'1} \wedge \bh^{'3} -\sinh(\alpha) \bh^{'2} \wedge \bh^{'3} \\ \cosh(\alpha) \bh^{'1} \wedge \bh^{'4} -\sinh(\alpha)  \bh^{'2} \wedge \bh^{'4} \end{array} \right]. \eeq 
  
\noindent The axial part of the tensor vanishes, while the vector part takes the form 
\beq {\bf V'} = [2 \cosh(\alpha) + \sinh(\alpha) \alpha ] \bh^{'1} + [2\sinh(\alpha) + \cosh(\alpha) \alpha] \bh^{'2}, \eeq

\noindent and the norm of ${\bf V'}$ is 
\beq |{\bf V'}|^2 = \alpha^2 -4. \eeq

\noindent We must consider three regions, where ${\bf V'}$ is timelike, null or spacelike. In each case we can adapt ${\bf V'}$ to be part of the frame basis. This restricts the isotropy group at zeroth order to be 3D at most. 

We will not compute the purely tensorial part of the torsion. However, ${t'}_{(ab)c}$ admits a 1D isotropy group consisting of spatial rotations about $\bh^{'2}$. Since ${\bf V'}$ lies in the span of the $\bh^{'1}$-$\bh^{'2}$ plane this property will be preserved, and so $dim~ H_0 = 1$. There is one functionally independent invariant at zeroth order, $t_0 = 1$. 

At first order, the isotropy group remains 1D, and no new functionally independent invariant appears. Thus, $dim~H_1 = 1$ and $t_1 = 1$, implying that the algorithm terminates. We conclude that the dimension of the symmetry group of this teleprallel geometry is $n = 1+4-1 = 4$ and the coframe, $\bh^a$, cannot be diffeomorphic to the coframe, $\bh^{'a}$.

\subsection{The PPGW metric}

We next consider an analogue teleparallel geometry arising from the subset of the vacuum PP-wave spacetimes admitting null rotations about $\bell$, known as the plane gravitational waves metrics in GR \cite{kramer}.  We consider the line element for the vacuum PPGW-wave metrics,
\beq ds^2 = - 2 du (dv + H(u,x,y) du) + 2 d\zeta d \bar{\zeta}, \label{ppwave} \eeq
\noindent where $\zeta = x + i y$ is a complex coordinate and $H$ satisfies $\Delta H = H_{,\zeta \bar{\zeta}} =0$. Thus for any complex-valued analytic function $f(\zeta, u)$: 
\beq H = f(u,\zeta) + \bar{f}(u,\bar{\zeta}). \nonumber \eeq
\noindent While for the sub-class of {\it plane gravitational waves} in GR, the complex-valued analytic function takes the form, 
\beq f(u, \zeta) = (A_0(u)+i A_1(u)) \zeta^2. \nonumber \eeq
\noindent To simplify matters, we will consider the transverse 2-space relative to real-valued coordinates, 
\beq ds^2 = - 2du (dv + H du) + dx^2 + dy^2, \nonumber \eeq
\noindent where $H$ takes the form
\beq H = 2A_0(u)( x^2 - y^2) - 4 A_1(u) xy. \nonumber \eeq

We will start with the complex null gauge: 
\beq \bh^a = \left[ \begin{array}{c} \bn \\ \bell \\ \bar{\bm} \\\ \bm \end{array} \right]= \left[ \begin{array}{c} dv + H du  \\ 
 dv  \\  \frac{1}{\sqrt{2}}(dx + i dy) \\  \frac{1}{\sqrt{2}}(dx - i dy) \end{array} \right] \eeq

\noindent and we will assume that the spin-connection is zero, $\omega^a_{~bc} = 0$. This corresponds to the inertial spin-connection for the null coframe which gives a consistent solution to the $f(T)$ field equations. This is consequently a proper frame; however, it is not an invariant frame as the frame has not been fully fixed using coordinate independent conditions from the torsion tensor and its covariant derivatives.

At zeroth order, we may apply a boost and spin defined in equations \eqref{1DBoost} and \eqref{1Dspin}, respectively, with the parameter defined by: \beq D^4 e^{2 i \theta} = m^a H_{,a} \label{PPGWBst}. \eeq
\noindent This will normalize the components of the torsion tensor so that
\beq \bT^a = \left [ \begin{array}{c} \bh^2 \wedge \bh^3 - \bh^2 \wedge \bh^4 \\ 0 \\ 0 \\ 0 \end{array} \right]. \eeq
\noindent However, in this resulting null frame, the spin-connection $\omega^a_{~bc}$ will no longer be trivial. The number of functionally independent invariants is zero, and the isotropy group consists of null rotations about $\bh^2 = \bell$. We note that the torsion trace vector $\bV$ and the torsion axial vector $\bA$ both vanish.

At first order by computing the covariant derivative of the torsion tensor, we may fix the two null rotation parameters in \eqref{2DNRIso}:
\beq \ell' = \bell,~~n' = \bn + B \bar{\bm} + \bar{B} \bm + |B|^2 \bell,~~\bm' = \bm + B \bell \nonumber \eeq
with the complex parameter,
\beq \bar{B} = \frac{ -n^b ln( m^a H_{,a})_{,b} (x+iy)}{2D e^{I\theta} }, \label{PPGWNR}\eeq

\noindent so that the only non-zero components relative to the null frame $\{ \bell, \bn, \bm, \bar{\bm} \} $ are:

\beq \begin{aligned} & T^1_{~23|3} =   \frac{2(A_0 - i A_1) (x- iy)^2}{\sqrt{2(A_0^2+A_1^2) (x^2 + y^2)^\frac32}}, \\ 
\end{aligned}  \eeq
\noindent and its complex conjugate $T^1_{~24|4}$. 

The frame is now an invariant frame as all of the parameters have been fixed: boost and spins by equation \eqref{PPGWBst}, null rotations about $\bell$ by \eqref{PPGWNR} and the null rotations about $\bn$ as the identity transformation. The spin-connection can be computed using the Lorentz transformation:

\beq \Lambda = \left[ \begin{array}{cccc} D & D^{-1} |B|^2 & B e^{i\theta} & \bar{B} e^{-i\theta} \\ 
0 & D^{-1} & 0 & 0 \\ 
0 & \bar{B} D^{-1} & e^{i\theta} & 0 \\
0 & B D^{-1} & 0 & e^{-\theta} \end{array} \right]. \label{PPGWtransf} \eeq

\noindent where $D, \theta$ and $B$ are respectively, the boost, spin and null rotation parameters given by equations \eqref{PPGWBst} and \eqref{PPGWNR}. The associated spin-connection has the following non-zero components:

\beq \begin{aligned} & \omega^1_{~12} = D^{-1} \bn' (D),~ \omega^1_{~13} = D^{-1} \bm' (D),~ \omega^1_{~14} = D^{-1} \bar{\bm}' (D), \\ 
&\omega^3_{~32} = -i \bn'(\theta),~\omega^3_{~33} = -i \bm'(\theta),\omega^3_{~34} = -i \bar{\bm}'(\theta),\\
&\omega^1_{~32} = -[\bn'(B) + D^{-1} B \bn'(D) - i B \bn'(\theta)],\\ 
& \omega^1_{~33} = -[\bm'(B) + D^{-1} B \bm'(D) - i B \bm'(\theta)], \\
& \omega^1_{~34} = -[\bar{\bm}'(B) + D^{-1} B \bar{\bm}'(D) - i B \bar{\bm}'(\theta)]. \end{aligned} \eeq

\noindent The primed indices denote the elements of the new Lorentz frame, ${\bh'}^a = \{ \bell', \bn', \bm', \bar{\bm}'\}$ and ${\bh'}^a = \Lambda^a_{~b} \bh^b$. This spin-connection will satisfy the constraints for a $f(T)$ teleparallel gravity solution; that is, it is anti-symmetric in the first two indices, has vanishing curvature and satisfies the anti-symmetric part of the $f(T)$ field equations. 

We note that the set of Cartan invariants so far, $\mathcal{T}^1$, consists of the components of the torsion tensor and its covariant derivatives. The first order Cartan invariants will be expressed in terms of the coordinates $x$ and $y$, the functions $A_0(u)$ and $A_1(u)$, and their frame derivatives. Thus, at first order there are 3 functionally independent real-valued scalars that can be constructed from the components of the covariant derivative of the torsion tensor (i.e.,  $t_1 = 3$) and the isotropy group at first order is zero (i.e., $dim~H_1 = 0$). 

At second order, there are no new functionally independent invariants and the second order isotropy group remains zero dimensional: therefore the algorithm terminates with $t_1 = t_2 =3$ and $dim~H_1 = dim~H_2=s=0$. The isometry group has dimension $r= 0+4-3 =1$. 

As the non-zero components of the torsion tensor and its covariant derivatives all have negative boost weight, according to \cite{Higher} this implies that all SPIs formed from these tensors must vanish and hence the PPGW teleparallel geometry has the $\text{VSI}_T$ property \cite{Higher}. This is similar to the analogous plane gravitational wave $\text{VSI}$ spacetimes in GR. However, there is one significant difference between the metric-based geometry and the teleparallel geometry; the torsion tensor and its covariant derivatives do not admit null rotation linear isotropy, and hence there will be no null rotation affine frame symmetry in this teleparallel geometry. Since a general torsion tensor can admit null rotations as an element of their linear isotropy, it is conceivable that there is a connection that will satisfy the anti-symmetric part of the $f(T)$ field equations for a plane wave geometry and preserve null rotation linear isotropy.

There are a number of interesting questions that arise from this example. Are there other teleparallel geometries outside of Minkowski space and the PPGW solution that are $\text{VSI}$ spacetimes? Is the set of Cartan invariants equivalent to the set of SPIs, $\mathcal{I}_T$, for all teleparallel geometries? Is $\text{VSI}_T$ equivalent to $\text{VSI}_{R}$? These questions will be considered in a forthcoming paper.
\newpage

\section{Discussion} \label{sec:discussion}

We have considered the equivalence problem for teleparallel geometries. These are manifolds equipped with a frame basis and a connection for which both of the curvature and non-metricity are zero; for such spaces, the torsion tensor and its covariant derivatives are the fundamental objects of study. The teleparallel geometries play an important role in alternative theories of gravitation, which are known as teleparallel theories of gravity where torsion replaces curvature as the invariant quantity describing gravity. 

To determine the equivalence of two teleparallel geometries, we have proposed a modification of the Cartan-Karlhede algorithm adapted to Riemann-Cartan geometries \cite{fon1992}. The algorithm relies on restricting the parameters of the Lorentz frame transformation group by fixing the torsion tensor and all covariant derivatives of the torsion tensor up to its $q$-th covariant derivative into canonical forms which is determined by the linear isotropy group of each tensor. At each iteration the number of functionally independent invariants and the remaining linear isotropy of the $q$-th covariant derivative is recorded. 
This process yields two invariant discrete sequences and a list of Cartan invariants that locally uniquely characterize the space. 

By identifying all of the permitted linear isotropy group of the torsion tensor and listing the canonical forms for the torsion tensor for each of the potential linear isotropy groups, we have restricted the linear isotropy group at higher iterations to be less than three. This result also imposes an upper bound on the number of iterations needed to classify any teleparallel geometry to be less than $q=7$.

In addition to classification, the proposed algorithm gives insight into the symmetry group of a teleparallel geometry. This is important since the metric is no longer the central object of study, and hence the set of isometries may not represent a set of intrinsic symmetries for a given teleparallel geometry. Introducing the more general concept of an affine frame symmetry as a diffeomorphism from the manifold to itself which leaves the frame basis determined by the Cartan-Karlhde algorithm and the associated spin-connection invariant, we note that the dimension of the affine frame symmetry group is provided by the algorithm. As any affine frame symmetry is an isometry but not all isometries are affine frame symmetries, this determines the dimension of the symmetry group for a teleparallel geometry. 

We note that the frame-dependent approach to calculate affine frame symmetries is equivalent to the method used in \cite{HJKP2018}; however, our approach can be used to establish more general results for teleparallel geometries. For example, in theorem \ref{thm:MaxSymFS} we have proven that the only teleparallel geometry which admits a maximal number of affine frame symmetries is Minkowski space. More generally, in corollary \ref{cor:maxsym} we have proven that {\it any} Riemann-Cartan geometry with non-zero torsion admits at most a seven-dimensional affine frame symmetry group.

In the six examples presented we have shown how the modified Cartan-Karlhede algorithm can be implemented and explicitly illustrate the relationship between the Cartan invariants, the group of affine  frame symmetries and the subgroup of isometries.  We have also shown that inequivalent solutions to the teleparallel field equations can give equivalent metrics. This suggests that for a given frame anzatz we can produce distinct teleparallel geometries by choosing different spin-connections which cannot be related to each other using Lorentz frame transformations. 

This observation motivates the search for other spin-connections which generate distinct torsion geometries with the same metric geometry. For example, in the torsion geometry based on the PPGW metric which lacks null rotation isotropy about the wave-vector $\ell_a$, it is possible to choose a different spin-connection for the given frame associated to the PPGW metric which will preserve the null rotation symmetry about $\ell_a$. We note that the  presented teleparallel geometry associated to the PPGW spacetime admits the $\text{VSI}$ property and that for any other choice of spin-connection which preserves the null rotation linear isotropy at all orders will necessarily be VSI as well. 

In future work we will examine the potential teleparallel geometries which arise from the so-called class of Kundt frames \cite{Higher}. In 4D Lorentzian geometry the associated metrics are contained in the class of degenerate Kundt metrics which cannot be characterized by their scalar polynomial curvature invariants. Therefore, we will examine the class of teleparallel geometries based on degenerate Kundt frames and investigate the subset which cannot be characterized by the set of SPIs, $\mathcal{I}_T$, formed from the torsion tensor and its covariant derivatives. 



\section*{Acknowledgments}
The authors would like to thank Jos\'{e} Pereira, Juri Obukhov,  Martin Kr\v{s}\v{s}\'{a}k and Christian B\"{o}hmer for their numerous communications on the subject material.  AAC was supported by the Natural Sciences and Engineering Research Council of Canada. RvdH is supported by the St Francis Xavier University Council on Research. DDM supported through the Research Council of Norway, Toppforsk grant no. 250367: Pseudo-
Riemannian Geometry and Polynomial Curvature Invariants: Classification, Characterisation and Applications.


\bibliographystyle{apsrev4-1}
\bibliography{Tele-Parallel-Reference-file0}

\newpage

\section*{Appendix: Lorentz frame transformations of the purely tensorial part of torsion} 

The transformation rules for the components of $t_{(ab)c}$ can be written compactly relative to a complex null frame.

\begin{itemize}

\item For a null rotation about $\ell$: 
\beq \bell' = \bell,~ \bn' = \bn + B \bar{\bm} + \bar{B} \bm + |B|^2 \bell,~ \bm' =  \bm + B \bell \label{lRot} \eeq

\noindent we have the following transformation of the components
\beq \begin{aligned}
t_{(13)1}' = & t_{(13)1}, \\
t_{(12)1}' = & \bar{B} t_{(13)1} + B t_{(14)1}+ t_{(12)1}, \\
t_{(33)1}' = & 2 B t_{(13)1}+2 t_{(33)1}, \\
t_{(14)3}' = & -2 \bar{B} t_{(13)1} + B t_{(14)1} + t_{(14)3},\\
t_{(23)1}' = & 2 |B|^2 t_{(13)1} + B^2 t_{(14)1} + 2 B t_{(12)1}  + 2 \bar{B} t_{(33)1} + t_{(23)1}, \\
t_{(13)2}' = & - |B|^2 t_{(13)1}- 2 B^2 t_{(14)1} - 3 B t_{(12)1}- \bar{B} t_{(33)1}-B t_{(14)3}+ t_{(13)2}, \\
t_{(22)1}'  = & \bar{B}^2 t_{(13)2}+ B^2 \bar{B} t_{(14)1}+2 |B|^2 t_{(12)1} + \bar{B}^2 t_{(33)1}+  B^2 t_{(44)1}+ \bar{B} t_{(23)1} \\ 
& + B t_{(24)1} + t_{(22)1}, \\
t_{(33)2}'  = & - B^3 t_{(14)1} - 2 B^2 t_{(12)1} - B^2 t_{(14)3}+ B t_{(23)1} + 2 B t_{(13)2}+ t_{(33)2}, \\
t_{(24)3}' =& - B\bar{B}^2 t_{(13)1}+2 B^2 \bar{B} t_{(14)1}+ 2 |B|^2 t_{(12)1}- \bar{B}^2 t_{(33)1}+ 2 B^2 t_{(44)1} \\ 
& + 2 |B|^2  t_{(14)3}-2 \bar{B} t_{(23)1}+3 B t_{(24)1} -2 \bar{B} t_{(13)2}+ 2 B t_{(14)2}+t_{(24)3}, \\
t_{(23)2}' =& - t_{(14)1} B^3 \bar{B}-2 t_{(12)1} B^2 \bar{B}- B^3 t_{(44)1}-B^2 \bar{B} t_{(14)3}+ |B|^2 t_{(23)1}  \\ 
& -2 B^2 t_{(24)1}+2 |B|^2 t_{(13)2}  -B^2 t_{(14)2}-B t_{(22)1} +t_{(33)2}\bar{B}-B t_{(24)3}+t_{(23)2}. \end{aligned} \eeq

\noindent The effect of a null rotation about $n$ can be determined by swapping $1 \leftrightarrow 2$, $B = \bar{C}$ and $\bar{B} = C$. 

\item For a boost and a spin, 
\beq \bell' = D \bell,~ \bn' = D^{-1} \bn, \bm' = e^{i\theta} \bm \label{BoostSpin} \eeq 
\noindent we find:
\beq \begin{aligned} 
t_{(13)1}' = & D^2 t_{(13)1},, \\
t_{(12)1}' = & D t_{(12)1}  , \\
t_{(33)1}' = & e^{2 i \theta} t_{(33)1}   , \\
t_{(14)3}' = & D t_{(14)3}  ,\\
t_{(23)1}' = & e^{i \theta} t_{(23)1}  , \\
t_{(13)2}' = & e^{i \theta} t_{(13)2}  , \\
t_{(22)1}'  = & D^{-1} t_{(22)1} , \\
t_{(33)2}'  = & D^{-1} e^{2 i \theta} t_{(33)2} , \\
t_{(24)3}' = & D^{-1} t_{(24)3} , \\
t_{(23)2}' = & D^{-2} e^{i \theta} t_{(23)2} .
\end{aligned} \eeq

\end{itemize}

\end{document}